\documentclass[acmtog,balance=false]{acmart}
\AtBeginDocument{%
  }

\usepackage{amsmath,amsthm,bm}
\usepackage{graphicx}
\usepackage{subcaption}
\usepackage{booktabs}
\usepackage{tabularx}
\usepackage{multirow}  
\usepackage{wrapfig}
\usepackage{xcolor}
\usepackage{float}  
\usepackage[most]{tcolorbox}
\usepackage{tikz}
\usepackage{algorithm}
\usepackage{algpseudocode}
\usepackage{pgfplots}
\pgfplotsset{compat=1.18}

\definecolor{figblue}{HTML}{72A8CF}
\definecolor{figorange}{HTML}{FFAD66}
\definecolor{figgreen}{HTML}{59A14F}
\definecolor{figpurple}{HTML}{B07AA1}
\definecolor{figbrown}{HTML}{9C755F}
\definecolor{figred}{HTML}{E15759}
\usepackage{siunitx}
\DeclareSIUnit[quantity-product = ]\percent{\char`\%}

\usetikzlibrary{positioning,calc,arrows.meta,fit,backgrounds}

\graphicspath{{figures/}}

\tcbset{%
  hlmath/.style 2 args={%
    enhanced, nobeforeafter, tcbox raise base, math upper,
    colback=#1!7, colframe=#1!55!black,
    coltitle=black,
    boxrule=0.5pt, arc=2.5pt,
    left=2pt, right=2pt, top=4pt, bottom=3pt,
    boxsep=1pt,
    equal height group=hlmathgrp,
    attach boxed title to top center={yshift=-3pt, yshifttext=-1pt},
    boxed title style={%
      enhanced, size=small, colback=#1!10, colframe=#1!55!black,
      boxrule=0.4pt, arc=1.5pt,
      left=1pt, right=1pt, top=-0.5pt, bottom=-0.5pt,
      width=15mm, halign=center,
    },
    title={\footnotesize\sffamily\color{black} #2},
  },
}
\newcommand{\hlmath}[3]{\tcbox[hlmath={#1}{#2}]{#3}}

\newtheorem{theorem}{Theorem}
\newtheorem{proposition}{Proposition}
\newcommand{\e}{\mathrm{e}}
\newcommand{\abs}[1]{\left|#1\right|}
\newcommand{\norm}[1]{\left\|#1\right\|}

\renewcommand{\j}{\mathrm{j}}
\newcommand{\ttf}[1]{\textbf{\textit{#1}}}

\setcopyright{cc}
\setcctype{by}
\copyrightyear{2026}
\acmYear{2026}
\acmJournal{TOG}
\acmVolume{45}
\acmNumber{6}
\acmArticle{199}
\acmMonth{12}
\acmDOI{10.1145/3842559}

\begin{document}

\title[Dirichlet Splatting: Differentiable Rendering for Wave-Based Inverse Problems]{Dirichlet Splatting:\texorpdfstring{\\}{ }
Differentiable Rendering for Wave-Based Inverse Problems}

\author{Xingyu Chen}
\orcid{0000-0002-1989-4935}
\affiliation{%
  \institution{University of California San Diego}
  \city{La Jolla}
  \state{California}
  \country{USA}
}
\email{xic063@ucsd.edu}

\author{Wuqiong Zhao}
\orcid{0000-0002-9550-7423}
\affiliation{%
  \institution{University of California San Diego}
  \city{La Jolla}
  \state{California}
  \country{USA}
}
\email{wqzhao@ucsd.edu}

\author{Xinyu Zhang}
\orcid{0000-0001-9688-8056}
\affiliation{%
  \institution{University of California San Diego}
  \city{La Jolla}
  \state{California}
  \country{USA}
}
\email{xyzhang@ucsd.edu}

\author{Tzu-Mao Li}
\orcid{0000-0001-5443-470X}
\affiliation{%
  \institution{University of California San Diego}
  \city{La Jolla}
  \state{California}
  \country{USA}
}
\email{tzli@ucsd.edu}

\begin{abstract}
We present \emph{Dirichlet Splatting}, a physically grounded
differentiable rendering method for inverse problems with wave-based
coherent sensors.
Unlike optical cameras, these sensors observe a scene in a
band-limited Fourier space whose continuous point response is
complex-valued and oscillatory.
Our key observation is that the point spread function of a reflector
in this Fourier space is the \emph{Dirichlet kernel}, not a Gaussian.
We instantiate each splat as a Dirichlet kernel in closed form, model
each primitive as a planar surfel rendered through a
\emph{surfel-aware Dirichlet splat} whose complex amplitude is set by
wave propagation physics, and introduce a sliding Frank--Wolfe
adaptation tailored to the rugged, non-convex landscape that the
Dirichlet oscillation induces.
On terahertz tomography, Dirichlet Splatting recovers reflector
geometry with an order of magnitude fewer parameters, faster
convergence, and higher accuracy than 3D Gaussian Splatting.
\end{abstract}



\begin{teaserfigure}
\centering
\begin{tikzpicture}
  \node[anchor=south west, inner sep=0] (teaserimg)
    {\includegraphics[width=\linewidth]{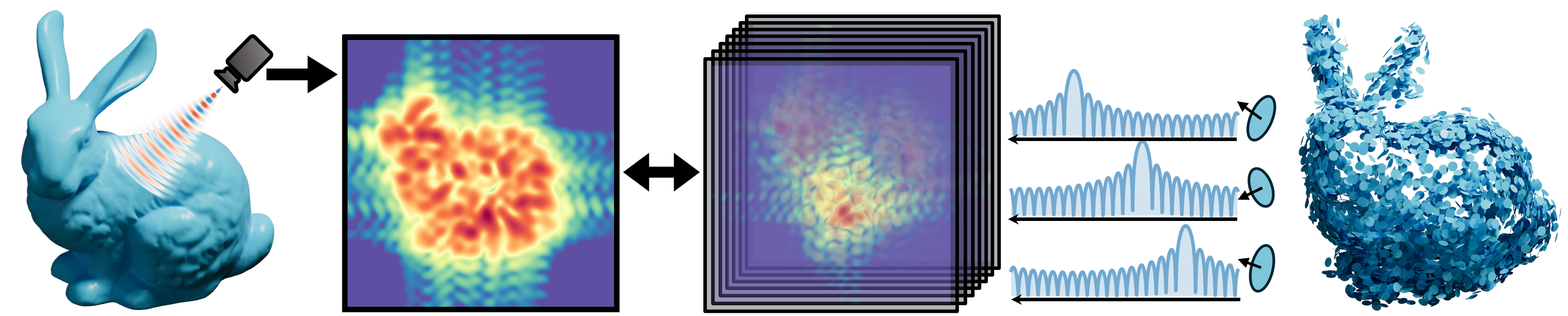}};
  \coordinate (dirtop) at ($(teaserimg.north west)!0.735!(teaserimg.north east)$);
  \node[anchor=south, font=\small\sffamily, inner sep=1pt, color=figred!70!black]
    at ([yshift=-2mm]dirtop) {Forward};
  \draw[-{Stealth[harpoon,swap,length=3mm,width=3mm]}, line width=0.7mm, color=figred]
    ([xshift=14mm,yshift=-4mm]dirtop) -- ([xshift=-14mm,yshift=-4mm]dirtop);
  \draw[-{Stealth[harpoon,swap,length=3mm,width=3mm]}, line width=0.7mm, color=figblue]
    ([xshift=-14mm,yshift=-5.8mm]dirtop) -- ([xshift=14mm,yshift=-5.8mm]dirtop);
  \node[anchor=north, font=\small\sffamily, inner sep=1pt, color=figblue!70!black]
    at ([yshift=-7.4mm]dirtop) {Backward};
  \begin{scope}[x={(teaserimg.south east)},y={(teaserimg.north west)}]
    \node[anchor=north, font=\large\sffamily]
      at (0.09, 0) {Physical World};
    \node[anchor=north, font=\large\sffamily]
      at (0.30, 0) {Spectrum};
    \node[anchor=north, font=\large\sffamily]
      at (0.54, 0) {Fourier Space};
    \node[anchor=north, font=\large\sffamily]
      at (0.73, 0) {Dirichlet};
    \node[anchor=north, font=\large\sffamily]
      at (0.91, 0) {Surfel Splats};
  \end{scope}
\end{tikzpicture}%
\caption{%
  \emph{Dirichlet Splatting} renders coherent wave-based measurements
  with a primitive whose spectral footprint is the physically exact
  Dirichlet kernel of a finite-window discrete Fourier transform.
  A coherent appearance model makes each surfel behave like a real
  surface element rather than an isotropic point reflector, while a
  tailored sliding Frank--Wolfe solver recovers reflector geometry
  from cold starts where Gaussian-splat baselines diverge.%
}%
\label{fig:teaser}
\Description{Teaser comparing Gaussian-splat and Dirichlet-splat
reconstructions of a generic coherent wave-based measurement, with an
inset showing Dirichlet vs Gaussian kernel structure on a dB scale.}
\end{teaserfigure}
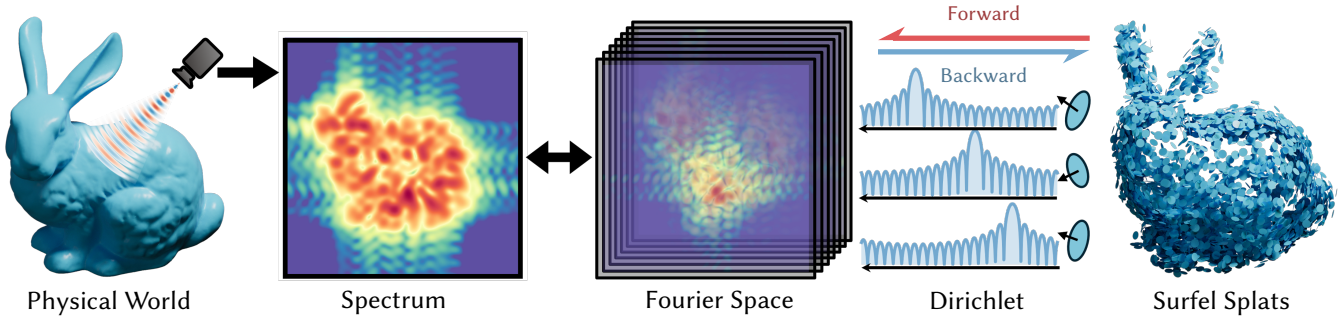

\maketitle

\section{Introduction}
\label{sec:intro}

Differentiable rendering has become a foundational tool for inverse
problems in graphics and sensing.
By making the forward image-formation pipeline differentiable
end-to-end, it turns scene reconstruction into gradient-based
optimization, recovering geometry, materials, lighting, and motion
directly from sensor measurements~\cite{mildenhall2020nerf,kerbl20233dgs}.
Nevertheless, \emph{non-visible-spectrum} sensors, including
terahertz (THz) tomographs, sonars, and radar, equally demand
differentiable rendering to solve their inverse problems.

3D Gaussian Splatting (3DGS)~\cite{kerbl20233dgs} renders scenes via
differentiable rasterization of anisotropic Gaussian primitives.
Recent work extends 3DGS to wave-based sensors via phase-coherent
combination of splats~\cite{wen2025wrfgs,yang2025gsrf}.
However, this port carries a structural mismatch, as the Gaussian
primitive, designed to approximate optical point spread, cannot
represent the complex-valued, oscillatory point response that
coherent wave measurement produces.

Optical imaging interpolates a smooth radiance field through a
near-Gaussian aperture, making a Gaussian splat a faithful local
approximation of how light from a small region accumulates in a
pixel.
Wave-based coherent imaging is structurally different.
A sparse scene of reflectors drives an analytic, continuous
physical operator, and the sensor records its values on a discrete
band-limited grid only at the readout stage of acquisition.
The natural target of an inverse solver is therefore the
continuous physics, namely the positions, orientations, and
amplitudes of the reflectors that drive this operator, rather than
the post-processed grid it produces.

\newlength{\kcintextsep}\setlength{\kcintextsep}{\intextsep}
\setlength{\intextsep}{8pt plus 2pt minus 2pt}%
\newlength{\kcwidth}\setlength{\kcwidth}{0.38\columnwidth}
\newlength{\kcpull}\setlength{\kcpull}{\dimexpr\columnsep-6pt\relax}
\begin{wrapfigure}[6]{r}{\dimexpr\kcwidth-\kcpull\relax}%
\vspace{-0.25\baselineskip}%
\setlength{\linewidth}{\kcwidth}%
\noindent\hspace*{-\kcpull}\hspace*{-0.3mm}%
\definecolor{kernelwarm}{HTML}{FAF5EB}
\definecolor{kernelcool}{HTML}{EFF7F9}
\begin{tikzpicture}[
  every node/.style={inner sep=0pt, outer sep=0pt},
  rowlabel/.style={anchor=north west, font=\footnotesize\sffamily\bfseries},
  collabel/.style={anchor=south, font=\footnotesize\sffamily},
]
  \def\hgap{0.45mm}
  \def\vgap{1.0mm}
  \def\boxpad{0.28mm}
  \def\labelpad{0.38mm}
  \pgfmathsetlengthmacro{\panelh}{(\linewidth - \hgap - 2*\boxpad)/3}
  \pgfmathsetlengthmacro{\curvew}{2.0*\panelh}

  \node (g1d) at (0, 0)
    {\includegraphics[width=\curvew]{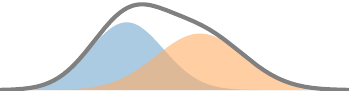}};
  \node[anchor=west] (g2d) at ([xshift=\hgap]g1d.east)
    {\includegraphics[height=\panelh]{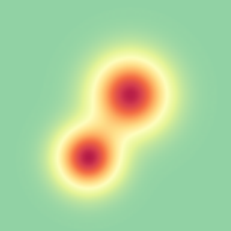}};

  \node[anchor=north west] (d2d) at ([yshift=-\vgap]g2d.south west)
    {\includegraphics[height=\panelh]{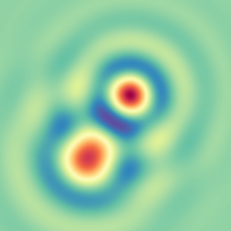}};
  \node (d1d) at ($(g1d |- d2d.center)+(0,-0.55mm)$)
    {\includegraphics[width=\curvew]{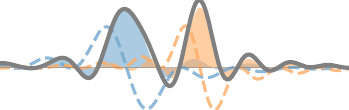}};

  \begin{scope}[on background layer]
    \node (gbox) [fit=(g1d)(g2d), fill=kernelwarm, rounded corners=2.2pt, inner sep=\boxpad] {};
    \node (dbox) [fit=(d1d)(d2d), fill=kernelcool, rounded corners=2.2pt, inner sep=\boxpad] {};
  \end{scope}

  \node[rowlabel] at ([xshift=\labelpad,yshift=-\labelpad]gbox.north west) {Gaussian};
  \node[rowlabel] at ([xshift=\labelpad,yshift=-\labelpad]dbox.north west) {Dirichlet};

\end{tikzpicture}%
\vspace{-0.7\baselineskip}
\end{wrapfigure}
\setlength{\intextsep}{\kcintextsep}
The point spread function (PSF) of this operator is complex-valued
and oscillatory (right inset), with sidelobes carrying analytic
information about the reflector geometry and a phase that governs
coherent interference between
reflectors~\cite{harris1978use}.
Existing ports of 3DGS approximate the operator by a smooth real
envelope on the discrete grid.
The approximation collapses both features and
reduces the inverse problem to a fit on a sampled output rather
than a differentiation of the underlying physics.

We present \emph{Dirichlet Splatting}, a 
differentiable rendering method for inverse problems with wave-based
coherent sensors.
\emph{First}, the PSF of a reflector in this
band-limited Fourier space is a \emph{Dirichlet
kernel}~\cite{harris1978use}.
We therefore instantiate each splat as a Dirichlet kernel evaluated
in $O(1)$ closed form.
\emph{Second}, we model each primitive as a planar \emph{surfel} and
render it through a \emph{surfel-aware Dirichlet splat} whose complex
amplitude is set by wave propagation physics, yielding a finite-area
surface element rather than an isotropic point reflector.
However, the very oscillation that makes the Dirichlet kernel
physically faithful induces a rugged, non-convex loss landscape.
Generic first-order optimizers such as Adam encounter spurious
gradient nulls in every sidelobe and stall before recovering even
simple multi-reflector scenes from cold starts.
\emph{Third}, we adapt \emph{Sliding Frank--Wolfe}~\cite{bredies2013inverse,denoyelle2020sliding}
to the Dirichlet landscape.
Our adaptation, \emph{Dirichlet Sliding Frank--Wolfe} (DSFW),
replaces the standard kernel-translation step with variable
projection over surfel parameters, residual dual certificates that
drive hard replacement of low-utility surfels, and periodic
low-resolution coupled Levenberg--Marquardt correction.

In summary, this paper makes the following contributions:

\ttf{(i)}
A \emph{Dirichlet splat primitive} whose spectral footprint is
  the exact point spread function in band-limited Fourier space
  (\S\ref{sec:method-kernel}).

\ttf{(ii)}
A \emph{surfel-aware coherent appearance model} that makes each
  splat behave as a real surface element rather than an isotropic
  point reflector (\S\ref{sec:method-surfel}).

\ttf{(iii)}
A sliding Frank--Wolfe adaptation tailored to the rugged
  Dirichlet loss landscape (\S\ref{sec:opt}).

\ttf{(iv)}
A real-data evaluation on sub-THz tomography, with ablations of the surfel
  modulation and the solver components, and convergence analysis (\S\ref{sec:results}).

\section{Related Work}
\label{sec:related}

\paragraph{Differentiable rendering and splatting primitives}
Differentiable rendering enables gradient-based optimization of scene
parameters from image observations, with NeRF-family volumetric
fields~\cite{mildenhall2020nerf,barron2022mipnerf360,mueller2022instant}
and differentiable rasterization~\cite{laine2020modular} as
representative families.
3DGS~\cite{kerbl20233dgs} introduced a real-time alternative based on projected 3D Gaussian primitives, since extended to dynamic scenes~\cite{luiten2024dynamic} and surface reconstruction~\cite{huang20242dgs}.
A common thread across these methods is that the choice of kernel
(i.e., Gaussian) is motivated by mathematical convenience rather than by the
physics of image formation. This is reasonable for optical cameras with
approximately Gaussian PSFs but structurally inappropriate for sensors
whose image formation is a finite-window discrete Fourier transform
(DFT).

In contrast, we preserve the splatting paradigm and its efficiency
while replacing the kernel with the physically exact Dirichlet
response of the underlying DFT.
A parallel line of work targets computer-generated holography (CGH)
by splatting complex-valued primitives directly onto a wavefront~\cite{choi2025gws,chao2025rpws,zhan2026cvhrf}.
These methods use complex-valued Gaussians to capture coherent
superposition (e.g., diffraction, interference) and demonstrate
state-of-the-art CGH quality.
However, equipping a Gaussian with a complex amplitude changes
\emph{what} the primitive carries, not the shape of the envelope
it splats. The envelope itself remains a smooth, monotonic Gaussian,
which cannot reproduce the periodic sidelobe nulls or the phase law
of a windowed DFT response (\S\ref{sec:motivation}).
The kernel mismatch is therefore structural rather than parametric,
and persists regardless of how rich the amplitude representation
becomes.
Our Dirichlet primitive instead replaces the envelope itself with
the analytic windowed-DFT response, embedding coherent sidelobe and
phase structure into the primitive rather than leaving them as
residuals for the optimizer to absorb.

\paragraph{Wave-optical rendering and appearance models}
A separate line of work makes the light transport itself wave-optical.
Generalized ray formulations carry phase, coherence, and diffraction
through the rendering integral and reduce to geometric optics in the
incoherent limit~\cite{steinberg2021generic,steinberg2024generalized},
while wave-optical appearance models replace the geometric-optics BRDF
with a diffraction-aware response derived from surface
microgeometry~\cite{stam1999diffraction,holzschuch2017two}.
These methods target passive imaging, where a scene is lit by ambient or
area sources and the sensor integrates a partially coherent field over a
finite aperture.
We instead target active coherent sensing, where the sensor illuminates
the scene with a known waveform and records the complex baseband return,
so that the measurement is a finite-window DFT whose point response is
the Dirichlet kernel in closed form (\S\ref{sec:motivation}).
A wave-optical renderer describes how a field arrives at an aperture,
while our appearance model describes what
a surfel contributes to a Fourier-space bin after that field has been
mixed, sampled, and transformed.

\paragraph{Differentiable coherent sensing}
Simulation of coherent wave-based sensors has traditionally relied on
ray-tracing methods that propagate electromagnetic or acoustic waves
geometrically and accumulate path losses, reflections, and
scattering~\cite{scheiner2020seeing,hirsenkorn2017virtual,auer2016ray,chen2025radio}.
Such simulators produce realistic time-domain signals but are
computationally expensive and difficult to differentiate.
Learning-based radar synthesis instead trains generative
models~\cite{wheeler2017deep,ouaknine2021carrada,chen2023rf} to bypass physical
simulation entirely, at the cost of physical grounding and gradient-based
inversion.
AutoMS and FlowForm use differentiable RF simulation to reshape mmWave coverage by
optimizing passive metasurfaces~\cite{ma2024automs,zhao2026flowform}.
A more direct route is to differentiate through the full waveform
pipeline (e.g., signal generation, mixing, sampling, fast Fourier
transform (FFT)) with automatic differentiation~\cite{chen2026rfdt}. This is physically
exact but materializes the full time-domain signal and its FFT
gradient at $O(MN\log MN)$ per step.
Recent neural-scene approaches such as DART~\cite{huang2024dart} and
Radar Fields~\cite{borts2024radarfields} adapt NeRF-style
representations to radar by combining implicit geometry with an
explicit or learned sensor model, but they either operate on
detection-level data (i.e., discarding spectral information) or approximate
the radar forward model with a neural network rather than its analytic
DFT response.
We instead render directly in the spectral domain using closed-form
Dirichlet responses, encoding the known DFT physics into the
rendering kernel itself rather than learning it from data.

\paragraph{Sparse recovery and variable projection}
That the DFT of a finite-length sinusoid is a Dirichlet
kernel~\cite{harris1978use,oppenheim1999discrete} and the related
windowed-FFT super-resolution practice~\cite{harris1978use} provide
the building blocks for our renderer. Off-the-grid sparse recovery
contributes sliding Frank--Wolfe with variable
projection~\cite{denoyelle2020sliding,jaggi2013revisiting,golub1973variable}
and subspace methods such as MUSIC~\cite{schmidt1986multiple} or
ESPRIT~\cite{roy1989esprit} bypass DFT resolution by exploiting
signal-subspace structure but do not provide gradients to a renderer.
Our optimizer combines sliding Frank--Wolfe, variable projection, and
residual dual certificates inside a differentiable splatting pipeline,
a combination that, to the best of our knowledge, has not previously
been applied to differentiable rendering.

\section{Background and Motivation}
\label{sec:motivation}

Wave-based coherent sensors form images by transmitting a known
waveform, mixing the received echoes against it, and applying a
finite-length DFT along each of $L$ measurement axes to assemble a
complex spectral tensor (Appendix~\ref{supp:theory}).
The axis count $L$ is modality-dependent (e.g., $L=1$ for coherent
spectroscopy, $L=2$ for synthetic-aperture imaging, $L=3$ for
range/Doppler/azimuth sensing).
The PSF-shaping step in the FFT-based imaging pipeline we target is a
finite-window DFT.
Each spectral bin records a complex amplitude, and contributions from
nearby reflectors interfere through their phase before any intensity
is observed.
Fig.~\ref{fig:physical-setup} sketches the source (a coherent
bistatic array), the targets (planar surfels with normal and
velocity), and the resulting Dirichlet PSF in the spectral tensor.

\input{figures/physical_setup}

A single point reflector at continuous bin offset $\Delta$ from a
spectral bin center contributes the exact analytic
\emph{Dirichlet kernel}
\begin{equation}
  \kappa_N(\Delta)
  = \e^{-\j\pi\frac{N-1}{N}\Delta}\,
    \frac{\sin(\pi\Delta)}{\sin(\pi\Delta/N)},
  \label{eq:dirichlet-kernel}
\end{equation}
where $N$ is the length of the DFT window applied by the sensor along
the corresponding measurement axis (i.e., the number of samples the
hardware coherently integrates before transforming).
Eq.~\eqref{eq:dirichlet-kernel} is the exact DFT response of a
windowed complex exponential, derived in
Appendix~\ref{supp:dirichlet-1d}.
The kernel is complex-valued, oscillatory, and periodic, with sidelobes
carrying \qty{\sim10}{\percent} of the total spectral energy for a
rectangular window~\cite{harris1978use} and a phase that governs
coherent interference between reflectors.
The sidelobe ratio and phase law are fixed by the window length $N$
and require no fitting.

A Gaussian splat (even with a complex amplitude) decays monotonically
and has no zero crossings, reproducing neither the periodic sidelobe
nulls nor the phase law that govern coherent interference between
reflectors at sub-bin separation. The resulting modeling error
dominates any reconstruction of coherent data.
The reflector-domain sparsity of wave-based scenes (most volume is
empty, with specular returns concentrating energy on isolated
surfaces) keeps a system-rank-equal Dirichlet representation
tractable, and adopting $\kappa_N$ as the rendering primitive matches
the windowed-DFT measurement physics
exactly~\cite{cetin2001feature}.

\section{Overview}
\label{sec:overview}

Dirichlet Splatting performs differentiable scene reconstruction
from active coherent measurements (THz, acoustic, or mmWave).
A known illuminator emits a band-limited waveform, and the sensor
records the complex echo as an $L$-axis spectral tensor whose axes
are modality-specific (e.g., depth alone for tomographic sensing;
range, Doppler, and azimuth for transceiver-array sensing; frequency
alone for coherent spectroscopy).
The scene is modeled as oriented planar surfels that scatter the
incident waveform once before reaching the sensor. Unlike the
emissive primitives of standard 3D Gaussian Splatting (3DGS), our surfels
are reflectors whose appearance is governed by complex reflectivity,
bistatic foreshortening, and a material-dependent Fresnel response.
Each surfel contributes a 3D Dirichlet kernel, the exact point-spread
function of the windowed DFT along each axis, placed directly in the
spectral domain rather than in the primal image domain of standard
splatting.

The renderer's projection, surfel-aware splats, and coherent compositing (\S\ref{sec:method}) feed DSFW (\S\ref{sec:opt}), an optimizer built on variable projection and residual dual certificates.      %

\section{Dirichlet Splatting}
\label{sec:method}

We render a scene of $S$ planar surfels into a coherent $L$-axis
spectral tensor in three logical steps.
Each surfel carries a small parameter set $\Theta_i$
(\S\ref{sec:method-primitives}). The renderer projects it to a
continuous Fourier coordinate $\mu_{i,\ell}$ on each measurement axis
$\ell$ via a sensor-specific map (\S\ref{sec:method-spectral}),
evaluates a surfel-modulated Dirichlet kernel whose width is set by
the FFT length (\S\ref{sec:method-surfel}), and coherently sums
across all surfels (\S\ref{sec:method-compositing}).
More specifically, the full forward model is
\begin{equation}\label{eq:method-surfel}
\hat{Y}(\mathbf{f})
=
\sum_{i=1}^{S}
\underbrace{%
\alpha_i \cdot
\hlmath{figblue}{am\smash{p}litude}{%
  \smash{\sqrt{A_i}\, V_i\, F_i\, R_i^{-\gamma}}\vphantom{()}%
}
\cdot
\hlmath{figorange}{kernel}{%
  \textstyle\smash{\prod_{\ell=1}^{L} d_{N_\ell}\!\big(f_\ell - \mu_{i,\ell}\big)}\vphantom{()}%
}%
}_{\mathrm{splat}_i(\mathbf{f})}\,,
\end{equation}
a coherent sum of per-surfel splats $\mathrm{splat}_i(\mathbf{f})$, each with a
complex amplitude (set by surfel physics) modulated by a
Dirichlet point spread (fixed by the per-axis FFT lengths
$N_\ell$).
The product form across $\ell$ is the exact separable DFT response of
a single point reflector
(Appendix~\ref{supp:dirichlet-separable}).
We evaluate the spectral tensor at the continuous Fourier coordinate
$\mathbf{f} = (f_1,\ldots,f_L)$, where $\ell$ indexes the measurement
axes. The measured bins lie at integer coordinates of $\mathbf{f}$.
Each kernel factor uses the fractional offset
$f_\ell - \mu_{i,\ell}$ between a bin and a surfel's continuous center.
In the $L=3$ running example,
$\mathbf{f} = (f_R,f_D,f_A)$ for range, Doppler, and azimuth.
The quantities $\alpha_i, A_i, V_i, F_i, R_i, \gamma$ are respectively
the complex reflectivity, oriented-disk area, bistatic foreshortening,
Fresnel response, sensor-to-surfel range, and path-loss exponent of surfel $i$
(\S\ref{sec:method-primitives}, \S\ref{sec:method-surfel}). The
per-axis Fourier center $\mu_{i,\ell}$ is the projection of
$(\mathbf{p}_i, \mathbf{v}_i)$ onto axis $\ell$
(\S\ref{sec:method-spectral}).

\subsection{Primitive Parameterization}
\label{sec:method-primitives}

We represent the scene as a collection of planar surfels, following the
geometric primitive of 2DGS~\cite{huang20242dgs} but replacing its
appearance model in full (\S\ref{sec:method-surfel}).

Each of the $S$ surfels carries the parameter set
\begin{equation}
\Theta_i = \{\mathbf{p}_i \in \mathbb{R}^3,\;
\mathbf{n}_i \in \mathbb{S}^2,\;
A_i \in \mathbb{R}_{>0},\;
\alpha_i \in \mathbb{C},\;
\mathbf{v}_i \in \mathbb{R}^3,\;
\eta_i\},
\label{eq:method-theta}
\end{equation}
where $\mathbf{p}_i$ is the surfel center, $\mathbf{n}_i$ its outward normal,
$A_i$ the area of the oriented disk it represents,
$\alpha_i$ a complex reflectivity
$\mathbf{v}_i$ a per-surfel translational
velocity that drives Doppler-axis behavior
(\S\ref{sec:method-spectral}; $\mathbf{0}$ for static scenes), and
$\eta_i \in \mathbb{C}$ the wave impedance of the surfel material,
entering the Fresnel response (\S\ref{sec:method-surfel}).
The geometric triple $(\mathbf{p}_i,\mathbf{n}_i,A_i)$ describes a small
oriented disk. The appearance pair $(\alpha_i,\eta_i)$ describes the
scene's coherent reradiation, distinct from 3DGS's spherical-harmonic
(SH) colors.
$\alpha_i$ is a scalar complex coefficient that scales the whole splat
and carries the surfel's magnitude and phase.
$\eta_i$ enters only through the Fresnel factor
$F_i \equiv F(\mathbf{n}_i,\theta_{\mathrm{inc},i},\eta_i)$ of
Eq.~\eqref{eq:method-surfel}, which controls the angular variation of
that magnitude (Eq.~\eqref{eq:method-fresnel}).
A Dirichlet surfel carries $13$ parameters versus $58$ for a 3DGS primitive ($\sim\!4.5\times$ fewer; Tab.~\ref{tab:primitive-comparison}).

\begin{table}[t]
\centering
\caption{\textbf{Primitive parameter comparison.}
Counts are degrees of freedom: a unit quaternion is three, a unit normal
two, and the complex amplitude is solved in closed form
(\S\ref{sec:opt-varpro}). Most of the reduction comes from appearance:
a coherent sensor uses one complex amplitude, whereas a camera uses a
directional radiance field.}
\label{tab:primitive-comparison}
\footnotesize
\setlength{\tabcolsep}{2pt}
\begin{tabular}{l l c l c}
\toprule
\textbf{Attribute} & \multicolumn{2}{c}{\textbf{3D Gaussian Splat}} &
                     \multicolumn{2}{c}{\textbf{Dirichlet Surfel (ours)}} \\
\cmidrule(lr){2-3}\cmidrule(lr){4-5}
            & expression & \#par & expression & \#par \\
\midrule
Position    & $\mathbf{p}\in\mathbb{R}^3$                          & 3
            & $\mathbf{p}\in\mathbb{R}^3$                          & 3 \\
Orientation & unit quaternion $\mathbf{q}\in\mathbb{H}$             & 3
            & unit normal $\mathbf{n}\in\mathbb{S}^2$               & 2 \\
Extent      & scale $\mathbf{s}\in\mathbb{R}^3_{>0}$                & 3
            & area $A\in\mathbb{R}_{>0}$                            & 1 \\
Velocity    & ---                                                   & 0
            & $\mathbf{v}\in\mathbb{R}^3$                           & 3 \\
Appearance  & SH coefficients, deg.\ $L=3$                       & 48
            & complex amplitude $\alpha\in\mathbb{C}$               & 2 \\
Opacity/material
            & opacity $o\in[0,1]$                                   & 1
            & impedance $\eta\in\mathbb{C}$                         & 2 \\
Kernel shape
            & Gaussian (from $\mathbf{q},\mathbf{s}$)               & 0
            & fixed by the aperture                                 & 0 \\
Compositing & alpha blending                                        & ---
            & coherent additive sum                                 & --- \\
\midrule
\textbf{Total / primitive} & & \textbf{58} & & \textbf{13} \\
\bottomrule
\end{tabular}
\end{table}

\subsection{Surfel-Aware Dirichlet Splat}
\label{sec:method-surfel}

The per-surfel splat $\mathrm{splat}_i(\mathbf{f})$ in
Eq.~\eqref{eq:method-surfel} is the complex contribution of surfel $i$
to $\hat{Y}(\mathbf{f})$.
Surfel attributes enter only as a pre-multiplicative complex amplitude
on the kernel product, cleanly separating the sensor model (the kernel
$\kappa_{N_\ell}$, evaluated through the unit-peak form
$d_{N_\ell}=\kappa_{N_\ell}/N_\ell$) from the scene model
($\alpha_i, A_i, V_i, F_i$). The Fourier coordinate $\mu_{i,\ell}$ is
defined per axis in \S\ref{sec:method-spectral},
Eq.~\eqref{eq:method-mu}.

\subsubsection{Complex Dirichlet kernel}
\label{sec:method-kernel}%
Each surfel's energy spreads across the Fourier-domain bins through
the Dirichlet kernel of \S\ref{sec:motivation}, the exact point-spread
function of an $N$-point DFT.
We parameterize it in \emph{bin units}: for a surfel with continuous
center $\mu$ evaluated at integer bin $k$, the bin offset is
$\Delta = k - \mu$.
We adopt the analytic Dirichlet form of Eq.~\eqref{eq:dirichlet-kernel}
(\S\ref{sec:motivation}) as the per-measurement-axis kernel (one
factor per axis $\ell$). It admits an $O(1)$ closed form per sample,
with a Taylor expansion for numerical stability near the removable
singularities at $\Delta \in N\mathbb{Z}$ where the denominator
vanishes (Appendix~\ref{supp:dirichlet-1d}; analytic gradient and
its singularity handling in Appendix~\ref{supp:gradients-kernel}).
The splat product uses the unit-peak variant
$d_N(\Delta) = \kappa_N(\Delta)/N$, with $d_N(0)=1$.
For sensors that apply a tapering window before the DFT, $d_N$
generalizes to the corresponding window's frequency response
(Appendix~\ref{supp:dirichlet-windowed}).%

\subsubsection{Coherent appearance model}
Each surfel reradiates a complex amplitude with magnitude from cell area,
foreshortening, Fresnel reflectivity, and free-space spreading, and phase
from round-trip propagation. Surfels combine coherently at the sensor and
are initialized from a random or preprocessed point cloud.

\paragraph{Amplitude factors}
Cell area enters as $\sqrt{A_i}$ so that $|\mathrm{splat}_i|^2$
contributes detected power proportional to $A_i$.
Foreshortening $V_i \equiv V(\mathbf{n}_i;\mathrm{tx},\mathrm{rx}) \in [0,1]$ is the bistatic geometric attenuation
\begin{equation}\label{eq:method-foreshortening}
V_i
=
\max(\hat{\mathbf{r}}_{\mathrm{tx},i}\!\cdot\!\mathbf{n}_i,0)\,
\max(\hat{\mathbf{r}}_{\mathrm{rx},i}\!\cdot\!\mathbf{n}_i,0),
\end{equation}
where $\hat{\mathbf{r}}_{\mathrm{tx},i}, \hat{\mathbf{r}}_{\mathrm{rx},i}$
are the unit-vector directions from transmitter and receiver to
surfel $i$ (\S\ref{sec:method-spectral}).
This factor is continuous and sub-differentiable in $\mathbf{n}_i$,
vanishing whenever the surfel is back-facing. %
Fresnel reflectivity $F_i \equiv F(\mathbf{n}_i,\theta_{\mathrm{inc},i},\eta_i)$ is the closed-form
angle-dependent response at a boundary of wave impedance $\eta_i$.
\begingroup\ignorespaces
Let $\theta_{\mathrm{inc},i}$ be the angle between the incident
direction and $\mathbf{n}_i$, and let $\eta_0$ be the wave impedance of
the surrounding medium.
For transverse-electric polarization at a planar boundary,
\begin{equation}\label{eq:method-fresnel}
F(\mathbf{n}_i,\theta_{\mathrm{inc},i},\eta_i)
=
\frac{\eta_i\cos\theta_{\mathrm{inc},i} - \eta_0\cos\theta_{t,i}}
     {\eta_i\cos\theta_{\mathrm{inc},i} + \eta_0\cos\theta_{t,i}},
\end{equation}
where the transmitted angle $\theta_{t,i}$ follows from Snell's law,
which for non-magnetic media reads
$\sin\theta_{t,i}/\sin\theta_{\mathrm{inc},i} = \eta_i/\eta_0$.
Intuitively, $\eta_i$ fixes the impedance mismatch at the boundary and
therefore the reflected fraction, while the two cosines make that
fraction vary with incidence angle.
A perfect conductor is the limit $\abs{\eta_i}\to 0$, where $F\to -1$ at
every angle. This is the regime of the metal targets in
\S\ref{sec:results-app}.
The transverse-magnetic coefficient and the derivation are in
Appendix~\ref{supp:fresnel}.
Path-loss exponent $\gamma\in[1,2]$ accounts for one-way or two-way
free-space spreading.
\par\endgroup\ignorespaces

\subsection{Spectral Rendering}
\label{sec:method-spectral}

Eq.~\eqref{eq:method-surfel} leaves two components of the forward
model unspecified: \textit{(i)}~a projection $\mu_i$ that supplies each
surfel's Fourier coordinates along every axis, and \textit{(ii)}~a coherent
compositing rule that sums the per-surfel splats into the spectral
tensor.

\subsubsection{Projection}
The role of $\mu_{\ell,i}$ mirrors the affine projection Jacobian in
3DGS.
Where 3DGS pushes a 3D covariance into a 2D image-space covariance,
we push a surfel's 3D position and velocity into Fourier-axis
offsets, one formula per axis $\ell$.
We instantiate the three formulas for the $L=3$ running example
below, with $\ell\in\{R,D,A\}$ for range, Doppler, and azimuth
(Fig.~\ref{fig:physical-setup}).

Let $\mathbf{o}$ denote the sensor origin and $\mathbf{u}_x$ the array axis.
For each surfel~$i$, we compute range $R_i = \norm{\mathbf{p}_i - \mathbf{o}}$,
the line-of-sight direction
$\hat{\mathbf{r}}_i = (\mathbf{p}_i - \mathbf{o})/R_i$,
radial velocity
$v_{r,i} = \langle \mathbf{v}_i, \hat{\mathbf{r}}_i \rangle$,
and azimuth
$\sin\theta_i = \langle \hat{\mathbf{r}}_i, \mathbf{u}_x \rangle$.
The continuous Fourier coordinates are then
\begin{equation}
\mu_{R,i} = \frac{N_R}{F_s}\frac{2K}{c}\,R_i, \quad
\mu_{D,i} = \frac{2N_D T_c}{\lambda}\,v_{r,i}, \quad
\mu_{A,i} = N_A\,\frac{\delta_{\mathrm{el}}}{\lambda}\sin\theta_i.
\label{eq:method-mu}
\end{equation}
The three coefficients map \textit{(i)}~round-trip delay $2R_i/c$ to range
bins, \textit{(ii)}~inter-pulse phase advance $4\pi v_{r,i}T_c/\lambda$ over
$N_D$ pulses to Doppler bins, and \textit{(iii)}~inter-element steering phase
$2\pi\delta_{\mathrm{el}}\sin\theta_i/\lambda$ over $N_A$ elements to
azimuth bins~\cite{richards2014fundamentals}.
Here $N_R, N_D, N_A$ are the per-axis FFT lengths and $\lambda$ the
carrier wavelength. The remaining constants ($F_s, K, T_c,
\delta_{\mathrm{el}}$) are sensor-side acquisition parameters whose
mapping is detailed in
Appendix~\ref{supp:impl}.
\paragraph{Doppler effects}
\label{sec:method-doppler}%
The radial component $v_{r,i}$ enters Eq.~\eqref{eq:method-mu} as the
Doppler-axis offset, and the same backward pass returns closed-form
Jacobians in both $\mathbf{p}_i$ and $\mathbf{v}_i$
(Appendix~\ref{supp:gradients-kernel}),
without the temporal extension layer of 3DGS-style dynamic
representations~\cite{luiten2024dynamic,yang2024deformable}.
The construction is exact under the standard coherent-sensing regime
(stop-and-hop, constant radial velocity, negligible range walk,
Doppler unambiguity). Validity bounds are in
Appendix~\ref{supp:fmcw-validity}.%

\subsubsection{Coherent compositing}
\label{sec:method-compositing}%
The complex spectral tensor is Eq.~\eqref{eq:method-surfel} evaluated
at the integer bin grid. For the $L=3$ running example,
$\hat{Y}[p,k,q] = \sum_{i=1}^{S} \mathrm{splat}_i(p,k,q)$ and the
observable power image is $\hat{H}[p,k,q] = \abs{\hat{Y}[p,k,q]}^2$.
Coherent signals superpose as complex phasors before the detector
squares the result, matching the true physics and correctly modeling
the constructive and destructive interference between reflectors
separated by a few resolution cells.
Compositing is purely additive (i.e., no depth order, no opacity),
making the renderer purely data-parallel.
When inter-surfel phase is not modeled, the renderer optionally
composites magnitude-squared splats via the Fej\'er kernel
(Appendix~\ref{supp:dirichlet-properties}).%

\section{Optimization}
\label{sec:opt}

The Dirichlet kernel of \S\ref{sec:method-kernel} is faithful in the
forward direction but oscillatory.
Its sidelobes propagate into the loss surface, yielding a rugged,
non-convex landscape on which first-order optimizers stall.

We collect the surfel parameters of \S\ref{sec:method} into a center matrix
$C \in \mathbb{R}^{S \times 3}$ and a complex coefficient vector
$\mathbf{b} = (\alpha_1,\dots,\alpha_S)^{\!\top} \in \mathbb{C}^S$ that
absorbs the per-splat amplitude product of Eq.~\eqref{eq:method-surfel}.
Let $\Omega$ index the measured spectral bins (so
$|\Omega|=N_R N_D N_A$ in the 3D case).
With $A(C) \in \mathbb{C}^{|\Omega| \times S}$ the matrix whose $i$-th
column is the spectral footprint of surfel $i$ and
$\mathbf{y}\in\mathbb{C}^{|\Omega|}$ the measured spectral image, the
inverse problem is
\begin{equation}
\min_{C,\,\mathbf{b}}\;
\mathcal{L}(C,\mathbf{b}) =
\frac{1}{|\Omega|}\norm{A(C)\mathbf{b} - \mathbf{y}}_2^2
+ \lambda \norm{\mathbf{b}}_2^2.
\label{eq:opt-objective}
\end{equation}
For reporting and for generic first-order baselines we additionally use a
log-power Charbonnier loss that handles 
the 40--60\,dB dynamic range typical of
coherent measurements~\cite{richards2014fundamentals}:
\begin{equation}
\mathcal{L}_{\mathrm{img}}
=
\frac{1}{|\Omega|}\sum_{k\in\Omega}
\sqrt{\bigl(\log(\hat H_k + \epsilon) - \log(H_k^{\mathrm{obs}} + \epsilon)\bigr)^2 + \epsilon_c^2},
\label{eq:opt-img-loss}
\end{equation}
with $\hat H_k = |[A(C)\mathbf{b}]_k|^2 + \sigma^2$ the rendered
power ($\sigma^2$ the noise floor),
$H_k^{\mathrm{obs}} = |y_k|^2$ the
measured power, $\Omega$ the set of evaluated bins, and constants
$\epsilon=10^{-6}$, $\epsilon_c=10^{-3}$.
Our Dirichlet Sliding-Frank-Wolfe (DSFW) algorithm minimizes Eq.~\eqref{eq:opt-objective} (closed-form in
$\mathbf{b}$). Eq.~\eqref{eq:opt-img-loss} is the outer evaluation
criterion and the loss used for AdamW baselines.

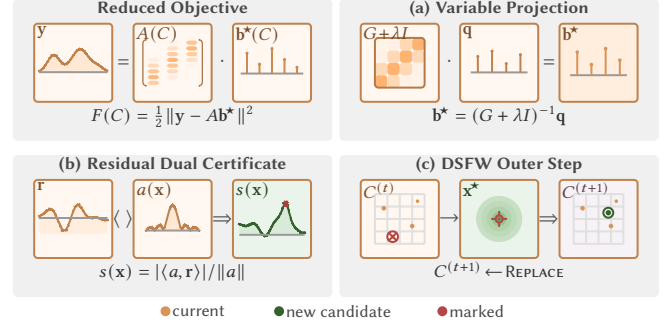
\begin{figure}[t]
\centering
\resizebox{\linewidth}{!}{\sffamily%
\begin{tikzpicture}[
  every node/.style={inner sep=0pt, outer sep=0pt},
  font=\footnotesize,
  iconbox/.style={
    rounded corners=1.4pt,
    draw=figorange!75!black, thick,
    fill=figorange!10,
    minimum width=11mm, minimum height=11mm,
  },
  intcap/.style={
    font=\fontsize{7pt}{7pt}\selectfont, color=figorange!45!black,
    anchor=north west, inner sep=0pt,
  },
  intcapG/.style={
    font=\fontsize{7pt}{7pt}\selectfont, color=figgreen!45!black,
    anchor=north west, inner sep=0pt,
  },
  intcapP/.style={
    font=\fontsize{7pt}{7pt}\selectfont, color=figpurple!50!black,
    anchor=north west, inner sep=0pt,
  },
  paneltitle/.style={
    font=\fontsize{7pt}{7pt}\selectfont\sffamily\bfseries, color=black!72,
    anchor=base,
  },
  paneleq/.style={
    font=\fontsize{7pt}{7pt}\selectfont, color=black!80,
    anchor=base,
  },
  eqsign/.style={font=\fontsize{8pt}{8pt}\selectfont, color=black!70,
                 inner xsep=0pt},
]

\def\pcL{-23mm}
\def\pcR{ 23mm}
\def\prT{  0mm}
\def\prB{-22mm}
\def\phw{22.5mm}
\def\phh{10mm}

\def\cofA{-14mm}
\def\cofB{  0mm}
\def\cofC{ 14mm}
\def\opAB{-7mm}
\def\opBC{ 7mm}

\def\rowyR{ 1mm}
\def\eqR  {-7.2mm}
\def\titleR{ 8mm}
\def\titleX{  0mm}

\begin{scope}[on background layer]
  \fill[black!6, rounded corners=2.0pt]
    ($(\pcL,\prT)+(-\phw,-\phh)$) rectangle ($(\pcL,\prT)+(\phw,\phh)$);
  \fill[black!6, rounded corners=2.0pt]
    ($(\pcR,\prT)+(-\phw,-\phh)$) rectangle ($(\pcR,\prT)+(\phw,\phh)$);
  \fill[black!6, rounded corners=2.0pt]
    ($(\pcL,\prB)+(-\phw,-\phh)$) rectangle ($(\pcL,\prB)+(\phw,\phh)$);
  \fill[black!6, rounded corners=2.0pt]
    ($(\pcR,\prB)+(-\phw,-\phh)$) rectangle ($(\pcR,\prB)+(\phw,\phh)$);
\end{scope}

\node[paneltitle]
  at ($(\pcL,\prT)+(\titleX,\titleR)$) {Reduced Objective};

\node[iconbox] (yIcon) at ($(\pcL,\prT)+(\cofA,\rowyR)$) {%
  \tikz[scale=0.22, baseline=-0.7mm]{%
    \fill[figorange!28]
      plot[smooth, samples=28, domain=-2.2:2.2]
      ({\x},{1.45*exp(-(\x-0.45)*(\x-0.45)*1.6)
            +1.05*exp(-(\x+1.25)*(\x+1.25)*3.0)
            +0.30*exp(-(\x-1.6)*(\x-1.6)*4.0)})
      -- (2.2,0) -- (-2.2,0) -- cycle;
    \draw[figorange!75!black, thick]
      plot[smooth, samples=28, domain=-2.2:2.2]
      ({\x},{1.45*exp(-(\x-0.45)*(\x-0.45)*1.6)
            +1.05*exp(-(\x+1.25)*(\x+1.25)*3.0)
            +0.30*exp(-(\x-1.6)*(\x-1.6)*4.0)});
    \draw[black!40, thick] (-2.2,0) -- (2.2,0);
  }%
};
\node[intcap] at ([xshift=0.4mm,yshift=-0.3mm]yIcon.north west) {$\mathbf{y}$};

\node[eqsign] at ($(\pcL,\prT)+(\opAB,\rowyR)$) {$=$};

\node[iconbox] (AIcon) at ($(\pcL,\prT)+(\cofB,\rowyR)$) {%
  \tikz[scale=0.27, baseline=-0.5mm]{%
    \foreach \cIdx/\py in {0/-0.7, 1/0.1, 2/0.8}{%
      \pgfmathsetmacro{\cx}{-1.0 + \cIdx*1.0}
      \foreach \rIdx in {0,1,2,3,4,5,6,7}{%
        \pgfmathsetmacro{\ry}{-1.05 + \rIdx*0.30}
        \pgfmathtruncatemacro{\opc}{max(8, 90-80*abs(\ry-\py))}
        \fill[figorange!\opc]
          ({\cx-0.31},{\ry-0.11}) rectangle ({\cx+0.31},{\ry+0.11});
      }
    }
    \draw[figorange!65!black, line width=0.5pt]
      (-1.42, 1.20) -- (-1.55, 1.20) -- (-1.55, -1.20) -- (-1.42, -1.20);
    \draw[figorange!65!black, line width=0.5pt]
      ( 1.42, 1.20) -- ( 1.55, 1.20) -- ( 1.55, -1.20) -- ( 1.42, -1.20);
  }%
};
\node[intcap] at ([xshift=0.4mm,yshift=-0.3mm]AIcon.north west) {$A(C)$};

\node[eqsign] at ($(\pcL,\prT)+(\opBC,\rowyR)$) {$\cdot$};

\node[iconbox] (bIcon) at ($(\pcL,\prT)+(\cofC,\rowyR)$) {%
  \tikz[scale=0.22, baseline=-0.7mm]{%
    \draw[black!40, thick] (-2.0,0) -- (2.0,0);
    \foreach \x/\h in {-1.6/1.3, -0.8/0.6, 0.0/1.6, 0.8/0.9, 1.6/0.4}{%
      \draw[figorange!80!black, thick] (\x,0) -- (\x,\h);
      \fill[figorange!80!black] (\x,\h) circle (0.10);
    }
  }%
};
\node[intcap] at ([xshift=0.4mm,yshift=-0.3mm]bIcon.north west) {$\mathbf{b}^{\!\star}\!(C)$};

\node[paneleq] at ($(\pcL,\prT)+(0,\eqR)$)
  {$F(C)=\tfrac12\|\mathbf{y}-A\mathbf{b}^{\!\star}\|^2$};

\node[paneltitle]
  at ($(\pcR,\prT)+(\titleX,\titleR)$) {(a)\,\,Variable Projection};

\node[iconbox] (G) at ($(\pcR,\prT)+(\cofA,\rowyR)$) {%
  \tikz[scale=0.25, baseline=-0.5mm]{%
    \foreach \i in {0,1,2,3}{%
      \foreach \j in {0,1,2,3}{%
        \pgfmathtruncatemacro{\valp}{15+78*exp(-((\i-\j)*(\i-\j))/1.2)}
        \fill[figorange!\valp]
          ({-1.4+0.7*\i},{-1.4+0.7*\j}) rectangle ({-0.7+0.7*\i},{-0.7+0.7*\j});
      }
    }
    \draw[figorange!70!black, thick] (-1.4,-1.4) rectangle (1.4,1.4);
  }%
};
\node[intcap] at ([xshift=0.4mm,yshift=-0.3mm]G.north west) {$G\!+\!\lambda I$};

\node[eqsign] at ($(\pcR,\prT)+(\opAB,\rowyR)$) {$\cdot$};

\node[iconbox] (q) at ($(\pcR,\prT)+(\cofB,\rowyR)$) {%
  \tikz[scale=0.25, baseline=-0.5mm]{%
    \draw[black!40, thick] (-1.6,0) -- (1.6,0);
    \foreach \x/\h in {-1.2/0.9, -0.4/1.2, 0.4/0.5, 1.2/1.0}{%
      \draw[figorange!75!black, thick] (\x,0) -- (\x,\h);
      \fill[figorange!75!black] (\x,\h) circle (0.085);
    }
  }%
};
\node[intcap] at ([xshift=0.4mm,yshift=-0.3mm]q.north west) {$\mathbf{q}$};

\node[eqsign] at ($(\pcR,\prT)+(\opBC,\rowyR)$) {$=$};

\node[iconbox, fill=figorange!22] (bstar) at ($(\pcR,\prT)+(\cofC,\rowyR)$) {%
  \tikz[scale=0.25, baseline=-0.5mm]{%
    \draw[black!40, thick] (-1.6,0) -- (1.6,0);
    \foreach \x/\h in {-1.2/1.3, -0.4/0.6, 0.4/1.5, 1.2/0.85}{%
      \draw[figorange!85!black, thick] (\x,0) -- (\x,\h);
      \fill[figorange!85!black] (\x,\h) circle (0.10);
    }
  }%
};
\node[intcap] at ([xshift=0.4mm,yshift=-0.3mm]bstar.north west) {$\mathbf{b}^{\!\star}$};

\node[paneleq] at ($(\pcR,\prT)+(0,\eqR)$)
  {$\mathbf{b}^{\!\star}=(G+\lambda I)^{-1}\mathbf{q}$};

\node[paneltitle]
  at ($(\pcL,\prB)+(\titleX,\titleR)$) {(b)\,\,Residual Dual Certificate};

\node[iconbox] (resid) at ($(\pcL,\prB)+(\cofA,\rowyR)$) {%
  \tikz[scale=0.22, baseline=-0.5mm]{%
    \fill[figorange!22]
      plot[smooth, samples=28, domain=-2.2:2.2]
      ({\x},{0.55*sin(deg(2.6*\x))*exp(-\x*\x/6)
            +0.40*sin(deg(4.2*\x+0.9))*exp(-\x*\x/4)
            +0.10})
      -- (2.2,-1) -- (-2.2,-1) -- cycle;
    \draw[figorange!75!black, thick]
      plot[smooth, samples=28, domain=-2.2:2.2]
      ({\x},{0.55*sin(deg(2.6*\x))*exp(-\x*\x/6)
            +0.40*sin(deg(4.2*\x+0.9))*exp(-\x*\x/4)
            +0.10});
    \draw[black!40, thick] (-2.2,0) -- (2.2,0);
  }%
};
\node[intcap] at ([xshift=0.4mm,yshift=-0.3mm]resid.north west) {$\mathbf{r}$};

\node[eqsign] at ($(\pcL,\prB)+(\opAB,\rowyR)$) {$\langle\,\rangle$};

\node[iconbox] (atom) at ($(\pcL,\prB)+(\cofB,\rowyR)$) {%
  \tikz[scale=0.25, baseline=-0.5mm]{%
    \def\atomcurve{
      (-1.50, 0.00) (-1.40, 0.10) (-1.30, 0.17) (-1.22, 0.19)
      (-1.10, 0.12) (-1.00, 0.00) (-0.90, 0.15) (-0.80, 0.27)
      (-0.72, 0.31) (-0.60, 0.23) (-0.50, 0.00) (-0.40, 0.34)
      (-0.30, 0.73) (-0.20, 1.10) (-0.10, 1.36) ( 0.00, 1.45)
      ( 0.10, 1.36) ( 0.20, 1.10) ( 0.30, 0.73) ( 0.40, 0.34)
      ( 0.50, 0.00) ( 0.60, 0.23) ( 0.72, 0.31) ( 0.80, 0.27)
      ( 0.90, 0.15) ( 1.00, 0.00) ( 1.10, 0.12) ( 1.22, 0.19)
      ( 1.30, 0.17) ( 1.40, 0.10) ( 1.50, 0.00)
    }
    \fill[figorange!28] plot coordinates {\atomcurve}
      -- (1.5,0) -- (-1.5,0) -- cycle;
    \draw[figorange!75!black, thick] plot coordinates {\atomcurve};
    \draw[black!40, thick] (-1.5,0) -- (1.5,0);
  }%
};
\node[intcap] at ([xshift=0.4mm,yshift=-0.3mm]atom.north west) {$a(\mathbf{x})$};

\node[eqsign] at ($(\pcL,\prB)+(\opBC,\rowyR)$) {$\Rightarrow$};

\node[iconbox, fill=figgreen!12] (cert) at ($(\pcL,\prB)+(\cofC,\rowyR)$) {%
  \tikz[scale=0.23, baseline=-0.5mm]{%
    \fill[figgreen!22]
      plot[smooth, samples=32, domain=-2.0:2.0]
      ({\x},{0.45+0.30*sin(deg(3*\x))+0.20*sin(deg(5.2*\x+0.4))
            +1.45*exp(-(\x-0.85)*(\x-0.85)*9)})
      -- (2.0,0) -- (-2.0,0) -- cycle;
    \draw[figgreen!60!black, thick]
      plot[smooth, samples=32, domain=-2.0:2.0]
      ({\x},{0.45+0.30*sin(deg(3*\x))+0.20*sin(deg(5.2*\x+0.4))
            +1.45*exp(-(\x-0.85)*(\x-0.85)*9)});
    \fill[figred!85]
      (0.85,2.10) -- (1.00,1.90) -- (0.85,1.70) -- (0.70,1.90) -- cycle;
    \draw[figred!75!black, thick]
      (0.85,2.10) -- (1.00,1.90) -- (0.85,1.70) -- (0.70,1.90) -- cycle;
    \draw[black!40, thick] (-2.0,0) -- (2.0,0);
  }%
};
\node[intcapG] at ([xshift=0.4mm,yshift=-0.3mm]cert.north west) {$s(\mathbf{x})$};

\node[paneleq] at ($(\pcL,\prB)+(0,\eqR)$)
  {$s(\mathbf{x})=|\langle a,\mathbf{r}\rangle|/\|a\|$};

\node[paneltitle]
  at ($(\pcR,\prB)+(\titleX,\titleR)$) {(c)\,\,DSFW Outer Step};

\node[iconbox] (Cold) at ($(\pcR,\prB)+(\cofA,\rowyR)$) {%
  \tikz[scale=0.25, baseline=-0.5mm]{%
    \draw[black!12, thick] (-1.4,-1.4) grid[step=0.7] (1.4,1.4);
    \fill[figorange!85!black] (-0.9, 0.6) circle (0.13);
    \fill[figorange!85!black] ( 0.7,-0.4) circle (0.13);
    \fill[figorange!85!black] ( 1.0, 1.0) circle (0.10);
    \fill[figorange!85!black] (-0.4,-1.0) circle (0.12);
    \draw[figred!80!black, thick] (-0.4,-1.0) circle (0.30);
    \draw[figred!80!black, thick] (-0.65,-1.25) -- (-0.15,-0.75);
    \draw[figred!80!black, thick] (-0.65,-0.75) -- (-0.15,-1.25);
  }%
};
\node[intcap] at ([xshift=0.4mm,yshift=-0.3mm]Cold.north west) {$C^{(t)}$};

\node[eqsign] at ($(\pcR,\prB)+(\opAB,\rowyR)$) {$\to$};

\node[iconbox, fill=figgreen!12] (peak) at ($(\pcR,\prB)+(\cofB,\rowyR)$) {%
  \tikz[scale=0.25, baseline=-0.5mm]{%
    \foreach \r in {1.30, 1.00, 0.70, 0.40, 0.18}{%
      \pgfmathtruncatemacro{\opc}{14+30/(\r+0.2)}
      \fill[figgreen!\opc] (0.55,0.35) circle (\r);
    }
    \draw[figred!80!black, thick] (0.55,0.35) circle (0.20);
    \fill[figred!80!black] (0.55,0.35) circle (0.07);
    \draw[figred!80!black, thick]
      (0.55,0.35)+(-0.45,0) -- ++(-0.18,0);
    \draw[figred!80!black, thick]
      (0.55,0.35)+(0.18,0) -- ++(0.45,0);
    \draw[figred!80!black, thick]
      (0.55,0.35)+(0,-0.45) -- ++(0,-0.18);
    \draw[figred!80!black, thick]
      (0.55,0.35)+(0,0.18) -- ++(0,0.45);
  }%
};
\node[intcapG] at ([xshift=0.4mm,yshift=-0.3mm]peak.north west) {$\mathbf{x}^{\!\star}$};

\node[eqsign] at ($(\pcR,\prB)+(\opBC,\rowyR)$) {$\Rightarrow$};

\node[iconbox, fill=figpurple!10] (Cnew) at ($(\pcR,\prB)+(\cofC,\rowyR)$) {%
  \tikz[scale=0.25, baseline=-0.5mm]{%
    \draw[black!12, thick] (-1.4,-1.4) grid[step=0.7] (1.4,1.4);
    \fill[figorange!85!black] (-0.9, 0.6) circle (0.13);
    \fill[figorange!85!black] ( 0.7,-0.4) circle (0.13);
    \fill[figorange!85!black] ( 1.0, 1.0) circle (0.10);
    \fill[figgreen!60!black] (0.55, 0.35) circle (0.15);
    \draw[figgreen!60!black, thick] (0.55, 0.35) circle (0.28);
  }%
};
\node[intcapP] at ([xshift=0.4mm,yshift=-0.3mm]Cnew.north west) {$C^{(t+1)}$};

\node[paneleq] at ($(\pcR,\prB)+(0,\eqR)$)
  {$C^{(t+1)}\!\leftarrow\!\textsc{Replace}$};

\def\legY{-34mm}
\fill[figorange!85!black] (-24mm,\legY) circle (0.7mm);
\node[font=\footnotesize\sffamily, color=black!72,
      anchor=west, inner sep=0pt]
  at (-23mm,\legY) {current};
\fill[figgreen!60!black] ( -8mm,\legY) circle (0.7mm);
\node[font=\footnotesize\sffamily, color=black!72,
      anchor=west, inner sep=0pt]
  at ( -7mm,\legY) {new candidate};
\fill[figred!80!black]   ( 15mm,\legY) circle (0.7mm);
\node[font=\footnotesize\sffamily, color=black!72,
      anchor=west, inner sep=0pt]
  at ( 16mm,\legY) {marked};

\end{tikzpicture}%
}%
\caption{%
  \textbf{DSFW overview.}
  DSFW maintains a surfel set $C$ by iterating three steps over the
  reduced objective $F(C)$ (top-left).
  \textbf{(a)}~\emph{Variable Projection} eliminates the linear
  coefficients $\mathbf{b}^{\!\star}$ in closed form.
  \textbf{(b)}~The \emph{residual dual certificate} $s(\mathbf{x})$
  peaks where adding a splat most reduces the residual, bypassing
  sidelobe minima.
  \textbf{(c)}~The outer step \emph{hard-replaces} the lowest-utility
  surfel with $\arg\max_{\mathbf{x}} s(\mathbf{x})$.
  Icons depict the complex spectral tensors $\mathbf{y}$, $\mathbf{r}$
  as magnitude envelopes, $A(C)$ as a dictionary matrix with one
  Dirichlet splat per column, and $a(\mathbf{x})$ as the Dirichlet
  kernel of Fig.~\ref{fig:physical-setup}.%
}%
\label{fig:dsfw-overview}%
\Description{}
\end{figure}

\subsection{Optimization Challenges}
\label{sec:opt-challenges}
Three properties of $\mathcal{L}$ make Eq.~\eqref{eq:opt-objective}
hard for plain gradient descent.
\paragraph{Sidelobe minima}
A rectangular window gives $\kappa_N$ \qty{-13}{dB} sidelobes that
place a gradient null at every sidelobe peak. A center initialized
one main-lobe width from the true peak descends into the nearest
sidelobe and stalls, and plain AdamW from a cold start fails on the
benchmarks of \S\ref{sec:results-inverse} and \S\ref{sec:results-3d}.
\paragraph{Amplitude shrinkage}
Eq.~\eqref{eq:opt-objective} couples centers and amplitudes:
reducing $|\alpha_i|$ lowers the residual by trading positional
accuracy for coefficient magnitude, leaving small-amplitude splats
parked at sidelobe minima rather than at true scatterers.
\paragraph{Inter-atom coupling}
Splats within a kernel width interfere coherently, coupling their
center gradients. Adam's diagonal preconditioning misses the
off-diagonal blocks of
$\mathcal{H} = \mathrm{Re}(J_c^{\!\mathsf H} J_c)/|\Omega|$, so two
coupled splats cannot make the joint move that resolves their
interference.

DSFW addresses these three challenges in turn
(Fig.~\ref{fig:dsfw-overview}).
\emph{Variable projection} (\S\ref{sec:opt-varpro}) eliminates
$\mathbf{b}$ analytically and leaves only the centers $C$ to optimize,
removing amplitude shrinkage from the search.
The \emph{residual dual certificate} (\S\ref{sec:opt-cert}) globally
scores how much a new splat at each grid point would reduce the
residual, allowing the search to move beyond local sidelobe minima.
DSFW (\S\ref{sec:opt-dsfw}) combines both mechanisms by replacing the
lowest-utility splat with the certificate maximizer and periodically
taking a Gauss--Newton step to account for inter-splat coupling.

\subsection{Variable Projection}
\label{sec:opt-varpro}

Our forward model is linear in the splat amplitudes $\mathbf{b}$,
since surfels superpose coherently as a sum of complex amplitudes
(Eq.~\eqref{eq:method-surfel}).
This linearity enables both the closed-form coefficient solve
developed in this subsection and the SFW outer loop of
\S\ref{sec:opt-dsfw}.
3DGS lacks this property because alpha compositing renders the
image as a nonlinear function of the per-Gaussian opacity and
color~\cite{kerbl20233dgs}, which precludes a closed-form solve.
The obstruction lies in the compositing rule rather than the
kernel. Gaussian variants with a linear sum-of-splats forward model
(e.g., GSRF~\cite{yang2025gsrf}, WRF-GS+~\cite{wen2025wrfgs}) admit
the same machinery, which we evaluate as GS+SFW, the Gaussian cell of
the $2{\times}2$ kernel--optimizer ablation in
\S\ref{sec:results-ablations}.

For fixed centers $C$, Eq.~\eqref{eq:opt-objective} is a complex ridge
regression with closed-form minimizer
\begin{equation}\label{eq:opt-bstar}
\begin{aligned}
\mathbf{b}^{\!\star}(C) &= \big(G(C) + \lambda I\big)^{-1} q(C), \\
G(C) &= \tfrac{1}{|\Omega|}\,A(C)^{\!\mathsf H} A(C), \quad
q(C) = \tfrac{1}{|\Omega|}\,A(C)^{\!\mathsf H} \mathbf{y}.
\end{aligned}
\end{equation}
Substituting back yields the variable-projected reduced
objective~\cite{golub1973variable,oleary2013variable},
\begin{equation}\label{eq:opt-vp}
F(C) = \mathcal{L}\!\big(C,\mathbf{b}^{\!\star}(C)\big),
\end{equation}
which depends only on the center variables.
Eliminating $\mathbf{b}$ removes amplitude shrinkage as a degree of freedom and
yields a better-conditioned nonlinear search at each step.

\subsection{Residual Dual Certificates}
\label{sec:opt-cert}

With $\mathbf{b}^{\!\star}(C)$ fixed by Eq.~\eqref{eq:opt-bstar}, the
\emph{residual dual certificate}~\cite{denoyelle2020sliding,jaggi2013revisiting}
provides a global signal for where to move centers next:
\begin{equation}\label{eq:opt-cert}
s(\mathbf{x}) =
{\big|\big\langle a(\mathbf{x}),\;
\mathbf{y} - A(C)\mathbf{b}^{\!\star}(C)\big\rangle\big|}\,\big/\,
{\|a(\mathbf{x})\|},
\end{equation}
where $a(\mathbf{x})$ is the spectral footprint a hypothetical surfel placed at position $\mathbf{x}$ would produce.
Intuitively, peaks of $s$ identify locations where adding a splat
most reduces the residual. Equivalently, $s$ is the dual variable of
the off-the-grid sparse recovery problem of which
Eq.~\eqref{eq:opt-objective} is the primal.

\paragraph{Global signal vs.\ local gradient}
Unlike a local gradient, $s(\mathbf{x})$ is evaluated everywhere on a
candidate grid, and therefore bypasses the sidelobe minima of
\S\ref{sec:opt-challenges}.
Its cost is dominated by candidate-voxel correlations and grows quickly
with grid resolution.
A fused kernel reduces certificate evaluation from a per-experiment
to a per-step cost (Appendix~\ref{app:runtime}).%

\subsection{Dirichlet Sliding Frank--Wolfe}
\label{sec:opt-dsfw}

The certificate of \S\ref{sec:opt-cert} drives a global, sparse-recovery
update that bypasses the sidelobe landscape entirely.
DSFW is \emph{inspired by} the sliding Frank--Wolfe (SFW)
template~\cite{denoyelle2020sliding} for off-the-grid spike recovery
on the BLASSO problem, but deviates in one important way.
Classical SFW \emph{adds} a new splat at
$\arg\max_{\mathbf{x}} s(\mathbf{x})$ and then merges nearby splats,
growing the support monotonically.
DSFW instead operates at a fixed budget of $S$ splats and
\emph{replaces} the lowest-utility splat with the certificate's
argmax.
The fixed budget reflects the physics of coherent imaging.
Aperture and bandwidth bound the number of resolvable scatterers,
and unbounded growth would fit sidelobes and noise as spurious splats.
We set $S$ slightly above a back-projection sparsity estimate
(\S\ref{sec:results-3d}).
The trade-off is that the classical SFW guarantees~\cite{denoyelle2020sliding}
no longer apply, since they require both the BLASSO non-degenerate-source
condition and unconstrained splat addition.

Each DSFW outer step performs four operations.
(i)~Solve Eq.~\eqref{eq:opt-bstar} for $\mathbf{b}^{\!\star}(C)$.
(ii)~Evaluate the certificate Eq.~\eqref{eq:opt-cert} on a coarse
candidate grid.
(iii)~\emph{Hard-replace} the surfel of lowest utility with
$\arg\max_{\mathbf{x}} s(\mathbf{x})$, where utility ranks splats by
amplitude and by a coherent-cancellation efficiency that
down-weights splats already representable by their neighbors.
(iv)~Optionally refine with a Levenberg--Marquardt sub-step on the
spatial residual.
Algorithm~\ref{alg:dsfw} specifies the loop, schedules, and
optional steps.
\begingroup\ignorespaces

\paragraph{Utility score}
Step (iii) must select a splat for replacement. Amplitude alone
is insufficient because neighboring splats may reproduce a
large-amplitude splat without increasing the residual, whereas an
isolated small-amplitude splat may be indispensable. We therefore score
each surfel $j$ by
\begin{equation}\label{eq:opt-utility}
u_j
=
\tfrac12\,\frac{\abs{b_j^{\!\star}}}{\max_k \abs{b_k^{\!\star}}}
+
\tfrac12\,\mathrm{eff}_j,
\qquad
\mathrm{eff}_j
=
\frac{\norm{P_{A_{\setminus j}}^{\perp}\,\mathbf{a}_j}_2}
     {\norm{\mathbf{a}_j}_2}
\in [0,1],
\end{equation}
where $\mathbf{a}_j$ is the spectral footprint of surfel $j$ and
$P_{A_{\setminus j}}^{\perp}$ projects onto the orthogonal complement of
$\mathrm{span}\{\mathbf{a}_k\}_{k\neq j}$.
Intuitively, the first term measures relative amplitude. The
\emph{coherent-cancellation efficiency} measures how much of the
surfel's footprint lies outside the span of the other splats. A value of
$\mathrm{eff}_j$ near zero means that the other splats already span this
footprint. Replacing such a surfel changes the residual little and
reallocates its center to the certificate maximum.

\paragraph{Initialization}
Cold-start runs draw the $S$ centers uniformly at random from the
candidate volume, which is the box spanned by the sensor's range gate and
its azimuth field of view, and give every center the same initial
amplitude. Where a modality supplies a cheap signal-processing estimate,
the centers may instead be seeded from back-projection peaks or from
FFT-domain thresholding, as the multi-view fit of
\S\ref{sec:results-bunny} and the printed-object fit of
\S\ref{sec:results-thz} do. Per-experiment candidate grids and budgets
are listed in
Appendix~\ref{supp:impl-hyperparams}.%
\par\endgroup\ignorespaces

DSFW carries no convergence guarantee, but recovers ground-truth
geometry from random initialization on every benchmark in
\S\ref{sec:results-inverse}.

\begin{algorithm}[t]
\caption{DSFW outer loop.}
\label{alg:dsfw}
\begin{algorithmic}[1]
\Require measurement $\mathbf{y}$, initial centers $C_0$, kernel order
$N$, ridge weight $\lambda$, iteration budget $T$, low-resolution
schedule $\mathcal{T}_{\mathrm{lr}}$, candidate grid $\mathcal{X}$
\Ensure refined centers $C$, coefficients $\mathbf{b}^{\!\star}$
\State $C \gets C_0$
\For{$t = 1, \dots, T$}
  \State $\mathbf{b}^{\!\star} \gets
    \arg\min_{\mathbf{b}} \norm{\mathbf{y} - A(C)\mathbf{b}}^2
                          + \lambda\norm{\mathbf{b}}^2$
    \Comment{VarPro, Eq.~\eqref{eq:opt-bstar}}
  \State $s(\mathbf{x}) \gets$ certificate on $\mathcal{X}$ from residual
    $\mathbf{r} = \mathbf{y} - A(C)\mathbf{b}^{\!\star}$
    \Comment{Eq.~\eqref{eq:opt-cert}}
  \State $u_j \gets \tfrac12 \abs{b_j^{\!\star}}/\max_k\abs{b_k^{\!\star}}
                  + \tfrac12\,\mathrm{eff}_j$ for each $j$
    \Comment{Eq.~\eqref{eq:opt-utility}}
  \State $j^\dagger \gets \arg\min_j u_j$
  \State $\mathbf{c}_{j^\dagger}
    \gets \arg\max_{\mathbf{x} \in \mathcal{X}}\, s(\mathbf{x})$
    \Comment{hard replacement}
  \State \emph{(optional)}\; refine $C$ by a Levenberg--Marquardt
    sub-step on the spatial residual
  \If{$t \in \mathcal{T}_{\mathrm{lr}}$}
    \State propose $\Delta C$ from a coupled low-resolution LM step and
      accept it if the full-resolution loss decreases
  \EndIf
\EndFor
\State \Return $C$, $\mathbf{b}^{\!\star}$
\end{algorithmic}
\end{algorithm}

Furthermore, DSFW augments the outer step with a periodic Gauss--Newton correction
and four small stabilizers (window matching, kernel-continuation
curriculum, pairwise repulsion, coarse-to-fine schedule on $N$). The refinements and experiment settings are described in
Appendix~\ref{supp:impl}.
With these corrections, DSFW achieves the lowest final loss among the
optimizers we evaluate (\S\ref{sec:results-inverse}).

\section{Results and Application}
\label{sec:results}

We evaluate Dirichlet Splatting on controlled synthetic benchmarks
(\S\S\ref{sec:results-forward}--\ref{sec:results-ablations}) and on a real
hardware application (\S\ref{sec:results-app}).
Synthetic experiments use a DFT setup with $N=256$ range
bins, $M=64$ slow-time bins, and $A=8$ array elements unless noted.
Ground-truth signal synthesis and per-experiment protocols are in Appendix~\ref{app:extra}.
AdamW hyperparameters and DSFW defaults are
listed in Appendix~\ref{supp:impl-hyperparams}.

\subsection{Baselines}
\label{sec:results-setup}

We compare against published splatters and classical/physics
references. Internal ablations are deferred to
\S\ref{sec:results-ablations}.
\textbf{3DGS}~\cite{kerbl20233dgs} is the canonical Gaussian splatter;
\textbf{GSRF}~\cite{yang2025gsrf} and \textbf{WRF-GS+}~\cite{wen2025wrfgs}
are the state of the art in coherent Gaussian splatting for
radio-frequency imaging;
\textbf{back-projection} is the matched-filter coherent adjoint;
and \textbf{WaveformAD} differentiates through the full
synthesis$\to$FFT$\to|{\cdot}|^2$ pipeline with autograd, serving as a
physics oracle that bounds attainable accuracy.
All comparisons share scene parameters, learning rate, truncation
radius, splat budget, and loss function unless noted.
\S\ref{sec:results-inverse} fits targets rendered by the same
operator, which isolates the solver from measurement-model error.
\S\ref{sec:results-bunny} shares that operator but scores geometry
against a closed mesh. The other experiments score against independent
targets, namely the analytic
FFT of a point reflector (\S\ref{sec:results-forward}), raw range
profiles from a standard FMCW chain (\S\ref{sec:results-3d}), and
measured captures (\S\ref{sec:results-app}).

\subsection{Forward Model Validation}
\label{sec:results-forward}

The first experiment tests the physical claim of \S\ref{sec:method-kernel}.
The Dirichlet kernel \emph{is} the PSF of the finite-window
DFT. A Gaussian splat is at best a main-lobe approximation that discards the
sidelobe energy carrying coherent phase information.

\paragraph{Setup}
We place a single complex reflector at fractional sub-bin offsets
$\delta \in \{0.0, 0.1, \ldots, 0.5\}$ on a finite-window grid and compare
each method's rendered response against the corresponding FFT-domain
ground truth.
We sweep over rectangular, Hann, and Kaiser ($\beta=14$) windows to verify
that the windowed kernel form (Appendix~\ref{supp:dirichlet-windowed}) tracks
arbitrary tapers, and over 1D, 2D, and 3D grids to verify separability.
The experiment isolates the kernel choice cleanly, with no
optimization, no scene complexity, and only forward evaluation against
an exact analytic reference.

\begin{figure}[t]
\centering
\includegraphics[width=\columnwidth]{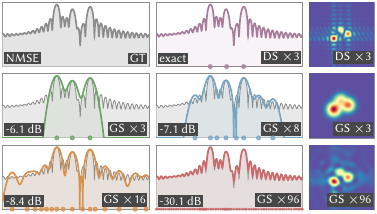}
\caption{%
  \textbf{Forward model fidelity.}
  Three Dirichlet primitives reproduce the sidelobe-rich coherent
  range profile of three complex reflectors exactly (top), whereas a
  fixed-width Gaussian representation still needs $96$ splats
  ($32{\times}$ more) to reach $-30.1$~dB magnitude normalized
  mean squared error (NMSE).
  Right: the corresponding 2D Dirichlet interference pattern in
  linear scale.
}%
\label{fig:forward-1d}%
\Description{Two rows of plots. The top row overlays a rendered range profile on an FFT ground truth in decibels for Dirichlet and Gaussian primitives at several splat counts. The bottom row plots normalized mean squared error against splat count. A right-hand panel shows a two-dimensional Dirichlet interference pattern in linear scale.}%
\end{figure}

\begin{figure*}[!tp]
\centering
\begin{tikzpicture}[
  , photo/.style  = {
    , inner sep = 0
    , outer sep = 0
    , anchor = south west
  }
  , mse label/.style = {
    , fill=black,fill opacity=0.62,text opacity=1,
        text=white,inner xsep=2pt,inner ysep=2pt
    , anchor = south west
    , outer sep = 0
    , font = \footnotesize\sffamily
  }
  , x = .1\linewidth
  , y = .1\linewidth
]
  \node at (0, 0) [photo] {\includegraphics[width=\linewidth]{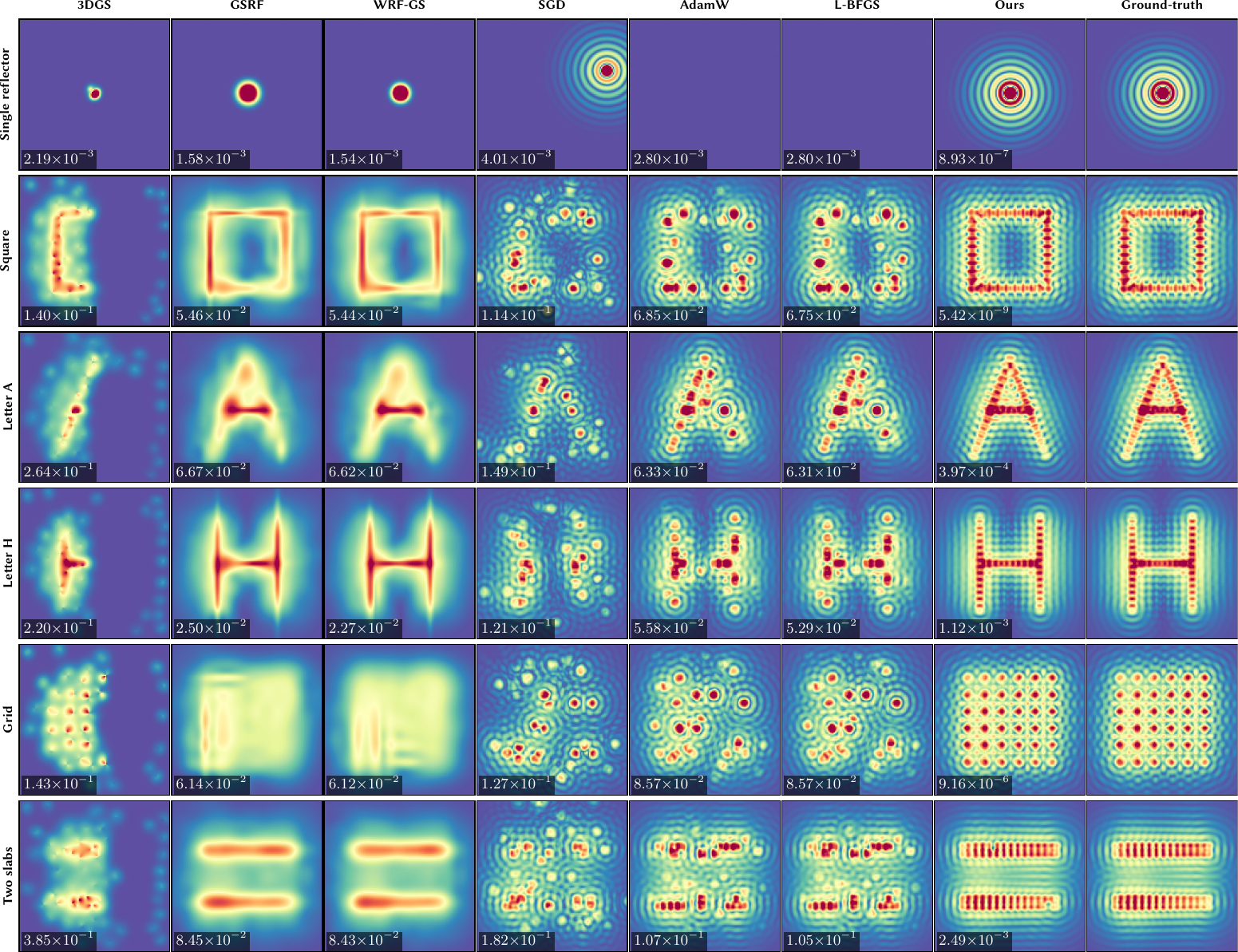}};

  \foreach \i in {0, ..., 5} {
    \node [mse label, yshift = 1pt] at (8.765, \i*1.261) {MSE};
  }
\end{tikzpicture}
\caption{%
  \textbf{Inverse fitting in signal space (\S\ref{sec:results-inverse}).}
  Methods fit Dirichlet splats directly to a 2D Fourier-space target
  tensor, with no aperture and no measurement model.
  Rows are interpretable signal-space targets (single reflector,
  square, letters A/H, grid, two slabs).
  Columns show a 3DGS-style no-phase baseline, GSRF, WRF-GS+, SGD,
  AdamW, an AdamW+L-BFGS quasi-Newton control, DSFW (ours), and the
  ground-truth scene, the noise-free Dirichlet-splat scene the
  targets were rendered from.
  Each panel shows the magnitude of the complex spectral tensor
  $\hat{Y}(\mathbf{f})$ of Eq.~\eqref{eq:method-surfel} over the two
  Fourier axes of the target.
  All methods share the same target tensor, cold initialization, splat
  budget, and loss function. Corresponding complex spatial MSE numbers
  are reported in \S\ref{sec:results-inverse}.
  Only DSFW recovers the sidelobe field. The Gaussian columns
  reconstruct the silhouette alone.}
\label{fig:inverse-result-grid}
\end{figure*}

\paragraph{Results}
Fig.~\ref{fig:forward-1d} shows the 1D comparison on a three-reflector
range profile with unequal amplitudes.
Three Dirichlet primitives reproduce the sidelobe-rich coherent waveform
directly, while a Gaussian representation must spend many splats
to approximate the same dB profile.
The right column shows the corresponding 2D Dirichlet interference pattern on
the same three-reflector configuration in linear scale.
With only three Gaussian splats, the fit captures the main peaks but misses
the sidelobe structure. Increasing the budget improves the envelope only
gradually, implying the Gaussian mismatch is a representation issue,
not just an optimizer issue.
Quantitatively, the three Dirichlet primitives are exact up to the numerical
floor (\qty{-120}{dB} NMSE), while Gaussian controls improve only slowly with splat
count: $3$, $8$, and $16$ splats remain between $-6.1$ and \qty{-8.4}{dB} NMSE, and
even $96$ splats reach only \qty{-30.1}{dB}.

\begingroup\ignorespaces
\paragraph{Doppler validation}
\label{sec:results-doppler}
The same kernel extends to the Doppler axis. Each surfel carries a
constant velocity that the rendering operator of
\S\ref{sec:method-spectral} projects into the slow-time spectrum.
We validate this projection on a $96{\times}64{\times}8$
frequency-modulated continuous-wave (FMCW) cube containing three
reflectors moving at $-6.0$, $0.0$, and $7.5$~m/s, with at most
$0.051$ range bins of range walk over the dwell.
The joint position--velocity Dirichlet model matches the finite-window
FFT oracle to the numerical floor (\qty{-120}{dB} complex NMSE) and picks
the correct velocity for every reflector from a \qty{0.05}{m/s} search
grid. Clamping velocity to zero on the same scene degrades the fit to
\qty{-1.11}{dB}, which isolates the contribution of the Doppler
parameter.
The experiment isolates the constant-velocity
parameterization. Dynamic-scene reconstruction remains outside its
scope.
Acquisition constants and the search protocol are in
Appendix~\ref{app:doppler-validation}.%
\par\endgroup\ignorespaces

\subsection{Inverse Fitting in Signal Space}
\label{sec:results-inverse}

This experiment tests the optimization claim of \S\ref{sec:opt}.
We scope the experiment to the \emph{signal level}, fitting Dirichlet
splats directly to a 2D Fourier-space target tensor with no aperture and
no measurement model.
The setup isolates the solver from any imaging error.
The full inverse imaging problem (raw measurements $\to$ geometry) is
studied in \S\ref{sec:results-3d}.
We compare DSFW against Dirichlet-AdamW and stochastic gradient
descent (SGD) optimizer controls and against
GSRF~\cite{yang2025gsrf} and WRF-GS+~\cite{wen2025wrfgs}
(\S\ref{sec:results-setup}). Optimizer-only and kernel-only ablations are
deferred to \S\ref{sec:results-ablations}.

\paragraph{Setup}
All targets in this section are synthetic 2D Fourier-space tensors of known Dirichlet-splat scenes, with methods fitting Dirichlet splats directly to them.
The convergence and basin studies use sparse-splat scenes:
the
random $K=10$ 2D scene, the structured $K=50$ star scene,
and the measured $K=5$ basin sweep.
The result grid (Fig.~\ref{fig:inverse-result-grid}) displays the same
signal-space pipeline on interpretable 
$128{\times}128$-bin spectral targets.
Splats are cold-initialized uniformly at random with no scene prior, and splat budget and loss follow \S\ref{sec:results-setup}.

\begin{figure}[t]
\centering
\includegraphics[width=\columnwidth]{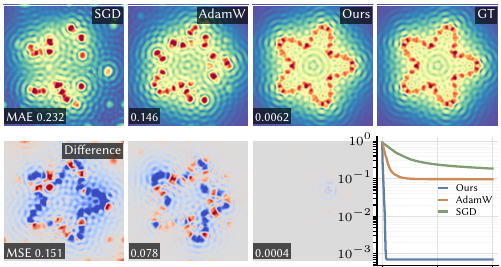}
\captionsetup{skip=4pt}
\caption{%
  \textbf{Cold-start convergence on a structured-star scene.}
  DSFW (ours) recovers the reflector field, while Dirichlet-AdamW
  stalls in sidelobe minima and SGD remains farther from the
  ground-truth scene under the same iteration budget.
  Top: magnitude reconstructions, labelled with mean absolute
  error (MAE).
  Bottom: complex spatial mean squared error (MSE) trajectories and
  signed magnitude-difference maps against the ground-truth scene, labelled
  with MSE.%
}%
\label{fig:convergence}%
\Description{A grid of magnitude reconstructions of a structured star scene for DSFW, Dirichlet-AdamW, and SGD, above loss-versus-iteration curves and signed magnitude-difference maps against the ground-truth scene.}
\vspace{-18pt}%
\end{figure}

\paragraph{Results}
Fig.~\ref{fig:convergence} shows a representative cold-start case.
DSFW recovers the reflector field and drives down the complex spatial loss, Dirichlet-AdamW lowers it into a mismatched sidelobe arrangement, and SGD lags farther under the same budget.
Across a $33{\times}33$ grid of cold-start offsets over 8 bins,
DSFW recovers all reflector positions to within $0.1$~bin on $77.4\%$ of trials,
against $0.07\%$ for Dirichlet-AdamW,
which stays trapped in the sidelobe neighborhoods of \S\ref{sec:opt-challenges}
(heat map in
Appendix~\ref{app:cross-scene-statistics}).
Fig.~\ref{fig:inverse-result-grid} extends the pattern to recognizable Fourier-domain targets.
Over the 7-scene benchmark, DSFW drives mean complex MSE from $3.32{\times}10^{-1}$ at initialization to $6.71{\times}10^{-4}$.
By contrast, Dirichlet-AdamW and SGD end at $6.06{\times}10^{-2}$ and
$1.12{\times}10^{-1}$, respectively.
The two-stage AdamW+L-BFGS quasi-Newton control reaches
$5.97{\times}10^{-2}$ after a $300$-step AdamW warmup and $50$ L-BFGS polish
iterations, giving only a small local refinement over AdamW and remaining two
orders of magnitude above DSFW.
The 3DGS-style no-phase baseline ends at $1.92{\times}10^{-1}$.
Its panels are visibly suppressed on one side of several letters and
slabs because it renders only a positive real Gaussian radiance field.
On those halves the complex Dirichlet target has strong imaginary
components or negative real projections, so the least-squares optimum
drives many splats toward low opacity rather than matching the displayed
magnitude.
GSRF~\cite{yang2025gsrf} and WRF-GS+~\cite{wen2025wrfgs} end at
$4.43{\times}10^{-2}$ and $4.38{\times}10^{-2}$. The Gaussian baselines
reconstruct smooth silhouettes but do not match the Dirichlet sidelobe field.
Cross-scene PSNR, magnitude NMSE, and SSIM are reported in
Appendix~\ref{app:cross-scene-statistics}.%

\subsection{Inverse Imaging from Raw Signals}
\label{sec:results-3d}

\begin{figure*}[!tp]
\centering
\includegraphics[width=\textwidth]{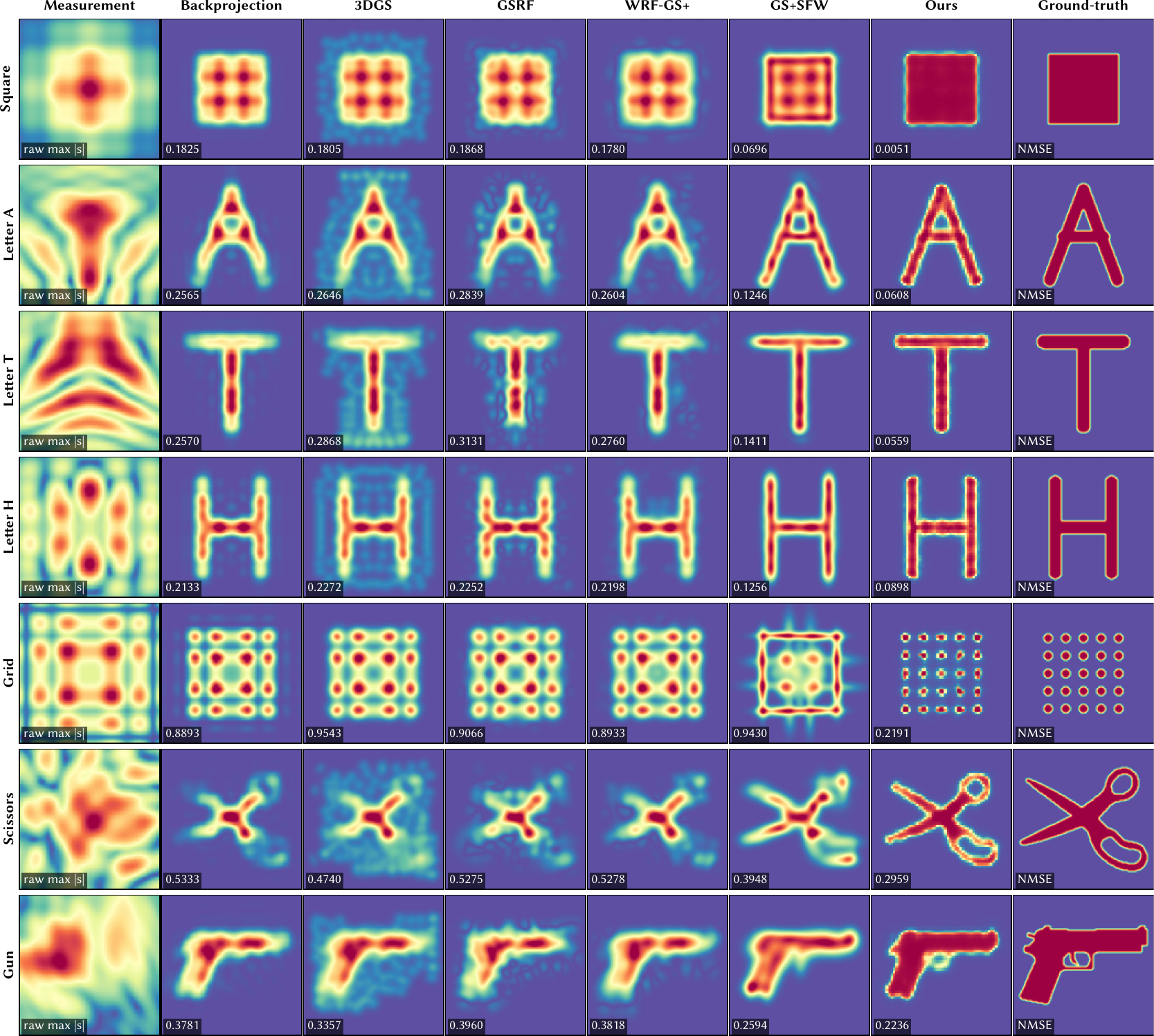}
\caption{%
  \textbf{Inverse imaging from raw signals (\S\ref{sec:results-3d}).}
  Each method fits raw complex multi-aperture range profiles of a target
  with known material response and decodes a dense geometric support
  estimate; reconstruction-panel labels report support-domain NMSE against
  the ground-truth support, meaning the binary occupancy map the
  simulator rendered from, with lower values better.
  Ours and GS+SFW are supervised on the raw range profiles and the
  remaining learned baselines on the back-projection image decoded from
  those same profiles. All learned methods share the cold initialization
  and the splat budget $S$.
  Rows are seven dense planar targets: a square, letters A/T/H, a reflector
  grid, a scissors silhouette, and a binarized pistol image.
  The leftmost column, {\emph{Measurement}}, is a measurement-domain sanity check (the simulator
  output collapsed to the per-aperture maximum range-profile magnitude),
  plotted in aperture coordinates rather than object coordinates, and is
  excluded from the support-domain NMSE labels.
  The remaining columns show back-projection, 3DGS, GSRF, WRF-GS+,
  GS+SFW, the raw-profile Dirichlet reconstruction (ours), and the
  ground-truth support.
  Only the raw-profile Dirichlet column keeps thin and separated
  components intact. The Gaussian columns merge them into a single
  blob.}%
\label{fig:known-plane-reconstruction}%
\Description{A seven-row by eight-column grid. Each row is a planar target and each column a method, from a measurement-domain sanity check through back-projection and four Gaussian splatting baselines to our raw-profile Dirichlet reconstruction and the ground-truth support. Panels carry support-domain error labels.}%
\end{figure*}

\S\ref{sec:results-inverse} fits splats to signal-space targets and
isolates the solver from the imaging measurement model.
We now evaluate the full inverse imaging problem: the input is a raw
complex multi-aperture scan of a target with known material response.
The output is a dense geometric support estimate, scored against ground-truth occupancy.

\paragraph{Setup}
The inverse-imaging grid uses 7 dense targets shown in Fig.~\ref{fig:known-plane-reconstruction}.
The simulator renders each target into raw complex range profiles using a
$96{\times}96$ aperture, \qty{4}{GHz} bandwidth, and a $56{\times}56$ candidate
grid.
The simulator follows the standard FMCW signal chain, synthesizing the
chirp echo from the target occupancy and material response, dechirping it
against the transmitted waveform, and taking the FFT of the beat signal.
Our model fits a Dirichlet support field directly to these raw profiles,
with eikonal, total variation, and binarization regularization.
Among baselines from \S\ref{sec:results-setup}, 3DGS, GSRF, and WRF-GS+ are
supervised on the back-projection image. \textbf{GS+SFW}, our coherent
Gaussian baseline sharing DSFW's forward model and SFW solver but with a
Gaussian kernel, is supervised on the raw range profiles.
All methods decode into a common occupancy space.

\paragraph{Results}
Reconstruction grids and per-scene support-domain NMSE are shown in Fig.~\ref{fig:known-plane-reconstruction}.
Ours reaches the lowest mean support NMSE across 7 targets ($0.136$), against $0.294$
(GS+SFW), $0.387$ (back-projection), $0.389$ (3DGS), $0.391$ (WRF-GS+),
and $0.406$ (image-domain GSRF).
Raw-channel fitting improves GSRF without closing the gap:
its raw NMSE stays an order of magnitude higher ($9.12{\times}10^{-3}$ vs.\ $7.74{\times}10^{-4}$ for ours),
with wider support-NMSE gaps on the grid ($0.943$ vs.\ $0.219$) and
scissors ($0.395$ vs.\ $0.296$).
GS+SFW is competitive on simple stroke-like shapes (e.g., letters A/T/H), which is consistent with broad Gaussian fields being effective
for smooth connected support but less reliable for separated or thin
multi-component geometry.
Cross-scene support NMSE, PSNR, SSIM, IoU, and Dice are reported in
Appendix~\ref{app:cross-scene-statistics}.%

\begingroup\ignorespaces
\subsection{3D Reconstruction}
\label{sec:results-bunny}

\begin{figure}
\newcommand{\bunnypanel}[1]{%
  \includegraphics[width=0.185\linewidth]{figures/bunny3d_view_#1.png}}%
\begin{tikzpicture}[
  x=0.193\columnwidth, y=0.193\columnwidth,
  panel/.style={inner sep=0, outer sep=0, anchor=north west},
  rlab/.style={font=\footnotesize\sffamily\bfseries, anchor=south, rotate=90,
               inner sep=0pt, yshift=1pt},
]
  \foreach \view [count=\i from 0] in {front, off0, off90, off180, off270} {
    \node[panel] (ours\i) at (\i, 0)  {\bunnypanel{ours_\view}};
    \node[panel] (ref\i)  at (\i, -1) {\bunnypanel{ref_\view}};
  }
  \node[rlab] at (-0.02, -0.5) {Ours};
  \node[rlab] at (-0.02, -1.5) {Ground truth};
\end{tikzpicture}%
\captionsetup{skip=4pt}
\caption{\textbf{3D reconstruction of the Stanford bunny
(\S\ref{sec:results-bunny}).}
We jointly fit one global surfel set to the raw complex range profiles
of five apertures clustered within \ang{25} of the front axis and render
it from the corresponding view directions. Color encodes depth along
the view axis using the colormap of Figs.~\ref{fig:inverse-result-grid}
and~\ref{fig:known-plane-reconstruction}. Each column uses the
ground-truth depth range for direct comparison. Our method recovers the
silhouette and ears from all five views.}
\label{fig:bunny3d}
\Description{Two rows of five panels each. The top row shows our
reconstruction of the Stanford bunny rendered as a depth map from five view
directions, and the bottom row shows the ground-truth surface rendered the same
way for comparison.}
\vspace{-12pt}%
\end{figure}
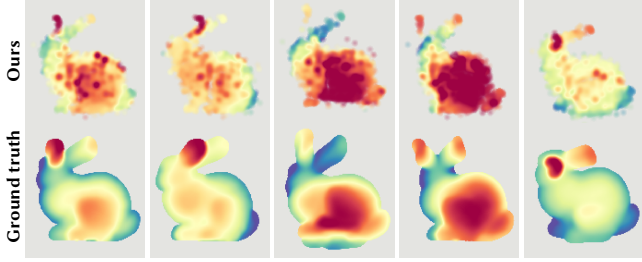

We evaluate Dirichlet Splatting on a 3D target, fitting
one global surfel set jointly across apertures that view it from
different directions.

\paragraph{Setup}
The target is a \qty{60}{mm} Stanford bunny. We sample its closed mesh
into $60{,}000$ oriented surfels and render each view with the forward
model of \S\ref{sec:method}. Each view uses a $128{\times}128$ virtual
aperture at \qty{2.5}{mm} pitch and \qty{350}{mm} standoff. The sensor
uses a \qty{120}{GHz} carrier, \qty{15}{GHz} bandwidth, and a
\qty{320}{mm} aperture, yielding \qty{1.4}{mm} resolution across the
target.
The fit runs against the concatenated raw complex range profiles of all
views under complex mean squared error. We concatenate the raw profiles
of the $V$ views into a single tensor. The foreshortening factor $V_i$ of
Eq.~\eqref{eq:method-foreshortening} provides per-view visibility, so
multi-view fitting reuses the forward model without a separate
visibility solver. We seed surfel centers from peaks of a joint
back-projection volume using the optional initializer of
\S\ref{sec:opt-dsfw}.

\paragraph{Metrics}
We measure (i) the median distance from each recovered surfel to the
ground-truth surface, (ii) F-score and recall at \qty{2}{mm}, (iii)
symmetric Chamfer distance, (iv) NMSE on the raw profiles, and (v) the
fraction of the ground-truth surface that at least one view observes.
Scores restricted to that fraction are marked \emph{recoverable}.
For a matched-budget comparison, we apply the same non-maximum
suppression to the brightest back-projection peaks and retain the same
number of points as our surviving surfels. We evaluate both point sets
with the same metrics.
Multi-view numbers are the mean $\pm$ half-range over three seeds.

\paragraph{Results}
Fig.~\ref{fig:bunny3d} shows the reconstruction from the five fitted view
directions. Three orthogonal views increase the observed fraction from
$0.35$ to $0.596$. They reach a full-surface recall@2mm of
$0.439 \pm 0.011$, a median surface distance of $1.67 \pm 0.07$~mm, and
a raw-profile NMSE of $0.203 \pm 0.012$. No single view reaches this
recall. Three views clustered within \ang{25} observe a more overlapping
surface region. Their raw-profile NMSE decreases to $0.069 \pm 0.012$,
and their recoverable F-score@2mm increases to $0.647 \pm 0.047$, at the
cost of lower surface coverage. Five clustered views reduce the median
surface distance to $1.49 \pm 0.21$~mm.
At a matched point budget on the three-view cluster, our reconstruction
reaches a median surface distance of \qty{1.51}{mm}, compared with
\qty{1.92}{mm} for back-projection. Its recoverable recall@2mm is $0.750$
versus $0.530$, and its symmetric Chamfer distance is \qty{3.87}{mm}
versus \qty{5.17}{mm}. Back-projection superposes energy without
enforcing cross-view consistency. Defocus caustics from individual views
therefore remain off the surface. It also produces an energy volume
rather than a model of the raw signal.

\par\endgroup\ignorespaces

\subsection{Ablations}
\label{sec:results-ablations}

We isolate the two main contributions, the Dirichlet kernel and the
DSFW solver, through a $2{\times}2$ factorial over kernel (Gaussian,
Dirichlet) and optimizer (AdamW, SFW), evaluated on a five-reflector
cold-start benchmark over $10$ seeds at a
shared splat budget and initialization, counting a trial successful when
every matched center lands within $0.1$~bin of the ground-truth scene.
The proposed pair succeeds on $100.0\%$ of trials at a final complex
spatial MSE of $5.32{\times}10^{-9}$.
\ttf{(i)} Holding the optimizer at SFW and replacing the Dirichlet kernel
with a Gaussian drops the success rate to $10.0\%$ and raises the
MSE from $5.32{\times}10^{-9}$ to $6.71{\times}10^{-3}$, a gap of more
than $10^6{\times}$, and reintroduces the ring artifacts of
\S\ref{sec:results-3d}.
Certificate-driven global search alone cannot compensate for a kernel
that discards sidelobe energy.
\ttf{(ii)} Holding the kernel at Dirichlet and replacing SFW with AdamW
reduces the cold-start success rate from $100.0\%$ to $0.0\%$
and raises the MSE to $1.18{\times}10^{-2}$,
confirming that the correct kernel is necessary but not sufficient on
the rugged landscape of \S\ref{sec:opt-challenges}.
Replacing both components with the standard Gaussian--AdamW
pairing yields a $0.0\%$ success rate and an MSE of
$8.54{\times}10^{-3}$.
Appendix~\ref{app:per-component-ablations} reports center RMSE.

\subsection{Application: Sub-THz Imaging}
\label{sec:results-app}
\label{sec:results-thz}

Sub-THz imaging penetrates optically opaque materials at quasi-optical resolution~\cite{willis2013infrastructs} and is exactly the finite-aperture-diffraction regime that motivated the Dirichlet kernel. We evaluate the full pipeline on real sub-THz measurements.

\paragraph{Setup}
We reuse the raw-signal fitting protocol of \S\ref{sec:results-3d} on
measurements from a custom-built sub-THz software-defined radio
(\qty{120}{GHz} center, \qty{1.23}{GHz} bandwidth, 8192 samples per
chirp).
We evaluate three metal letter scans: T, H, and Z.
We call a measured capture a \emph{raw scan}, as distinct from the
synthetic ground-truth support of \S\ref{sec:results-3d}.
The T scan uses a $37{\times}49$
\qty{2.5}{mm} raster; H and Z use $37{\times}37$ rasters. All use
50 coherent captures per aperture, a \ang{20} Tx/Rx bistatic
angle, and a \qty{15}{cm} focus distance.
We compare a
region-of-interest power projection (ROI Power), bistatic
back-projection, and GS+SFW from \S\ref{sec:results-3d}. Without a pixel-accurate ground-truth mask,
we report label-free contrast, support compactness, and raw-profile
fit residuals.

\begin{figure*}[!tp]
\centering
\begin{tikzpicture}
  \node[anchor=south west, inner sep=0] (thzimg)
    {\includegraphics[width=\textwidth,height=0.23\textheight,keepaspectratio]{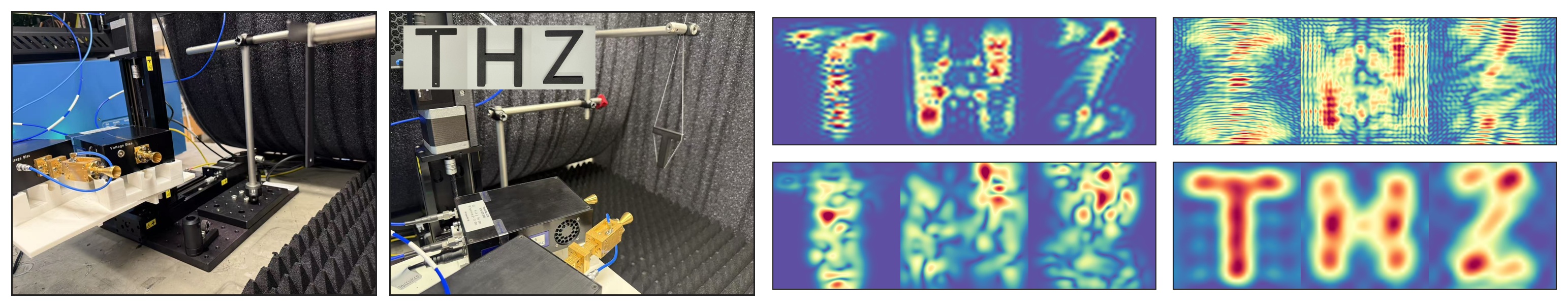}};
  \begin{scope}[x={(thzimg.south east)},y={(thzimg.north west)}]
    \tikzset{
      thzbox/.style={draw=cyan!85!blue,line width=0.55pt},
      thzboxhalo/.style={draw=white,line width=1.75pt},
      thzcallout/.style={draw=cyan!85!blue,line width=0.55pt,
        densely dashed,-{Latex[length=2.2pt,width=2.2pt]}},
      thzlabel/.style={fill=black,text=cyan!85!blue,inner sep=0.8pt,
        font=\scriptsize\sffamily\bfseries},
      thzgrouplabel/.style={fill=black,fill opacity=0.62,text opacity=1,
        text=white,inner xsep=2pt,inner ysep=0.8pt,
        outer sep=0,
        font=\scriptsize\sffamily\bfseries}
    }
    \draw[thzboxhalo] (0.035,0.341) rectangle (0.074,0.536);
    \draw[thzbox] (0.035,0.341) rectangle (0.074,0.536);
    \node[thzlabel,anchor=south west] at (0.035,0.545) {Tx};
    \draw[thzboxhalo] (0.087,0.407) rectangle (0.116,0.560);
    \draw[thzbox] (0.087,0.407) rectangle (0.116,0.560);
    \node[thzlabel,anchor=south west] at (0.087,0.568) {Rx};
    \draw[thzboxhalo,line width=2.0pt] (0.414,0.411) rectangle (0.44,0.603);
    \draw[thzbox,line width=0.78pt] (0.414,0.411) rectangle (0.44,0.603);
    \draw[thzcallout] (0.414,0.603) -- (0.380,0.709);
    \draw[white,line width=1.7pt] (0.257,0.709) rectangle (0.380,0.917);
    \draw[black,line width=0.55pt] (0.257,0.709) rectangle (0.380,0.917);
    \node[thzgrouplabel,anchor=south west] at (0.492,0.530) {(a) ROI Power};
    \node[thzgrouplabel,anchor=south west] at (0.748,0.530) {(b) Back Projection};
    \node[thzgrouplabel,anchor=south west] at (0.492,0.060) {(c) Gaussian};
    \node[thzgrouplabel,anchor=south west] at (0.748,0.060) {(d) Ours};
  \end{scope}
\end{tikzpicture}
\caption{%
  \textbf{Real sub-THz reconstruction.}
  The imaged scene consists of three metal letters T, H, and Z suspended
  in front of a sub-THz transceiver pair (left: setup photograph and a
  zoomed view of the hanging target).
  Panels (a)--(d) compare four reconstructions on the same measurement:
  (a) region of interest power, (b) back projection, (c) Gaussian, and
  (d) ours.
  All four methods decode the same raw scan on a common pixel
  grid and color scale. Our reconstruction preserves the three separated
  letters while reproducing the measured finite-window sidelobes.}%
\label{fig:thz}%
\Description{A photograph of a sub-terahertz transceiver pair with the transmit and receive probes outlined, a zoomed view of three metal letters hanging in front of them, and four reconstructions of those letters from the same raw scan: region-of-interest power, back projection, a Gaussian fit, and ours.}%
\end{figure*}
\paragraph{Results}
Fig.~\ref{fig:thz} and the three-scan baseline audit show distinct
failure modes.
ROI Power is a strong aperture-coordinate energy map
(average $p_{99}/p_{50}$ contrast $44.8$, robust SNR $30.5$),
yet not a geometric inversion:
at a $35\%$-peak threshold it occupies only $9.6\%$ of pixels across $10.3$ disjoint components.
Back-projection brings in the bistatic geometry at the cost of diffraction sidelobes: with $p_{99}/p_{50}$ falling to $3.71$
and $35\%$-peak support spreading over $26.2\%$ of pixels in $69.0$ components.
GS+SFW reaches mean normalized loss $1.70$ ($1.95$/$1.42$/$1.72$ for T/H/Z) on a compact
support ($10.7\%$ area, $8.3$ components) of broad blobs that miss the measured finite-window sidelobe structure.

\begin{figure}
\begin{tikzpicture}[
  panel/.style={inner sep=0, outer sep=0, anchor=south west},
  plab/.style={fill=black, fill opacity=0.62, text opacity=1, text=white,
               inner xsep=2pt, inner ysep=0.8pt, outer sep=0,
               anchor=south west, font=\scriptsize\sffamily\bfseries},
  , x = .1\linewidth
  , y = .1\linewidth
]
  \node[panel] (setup) at (0, 0)
    {\includegraphics[height=.42\linewidth]{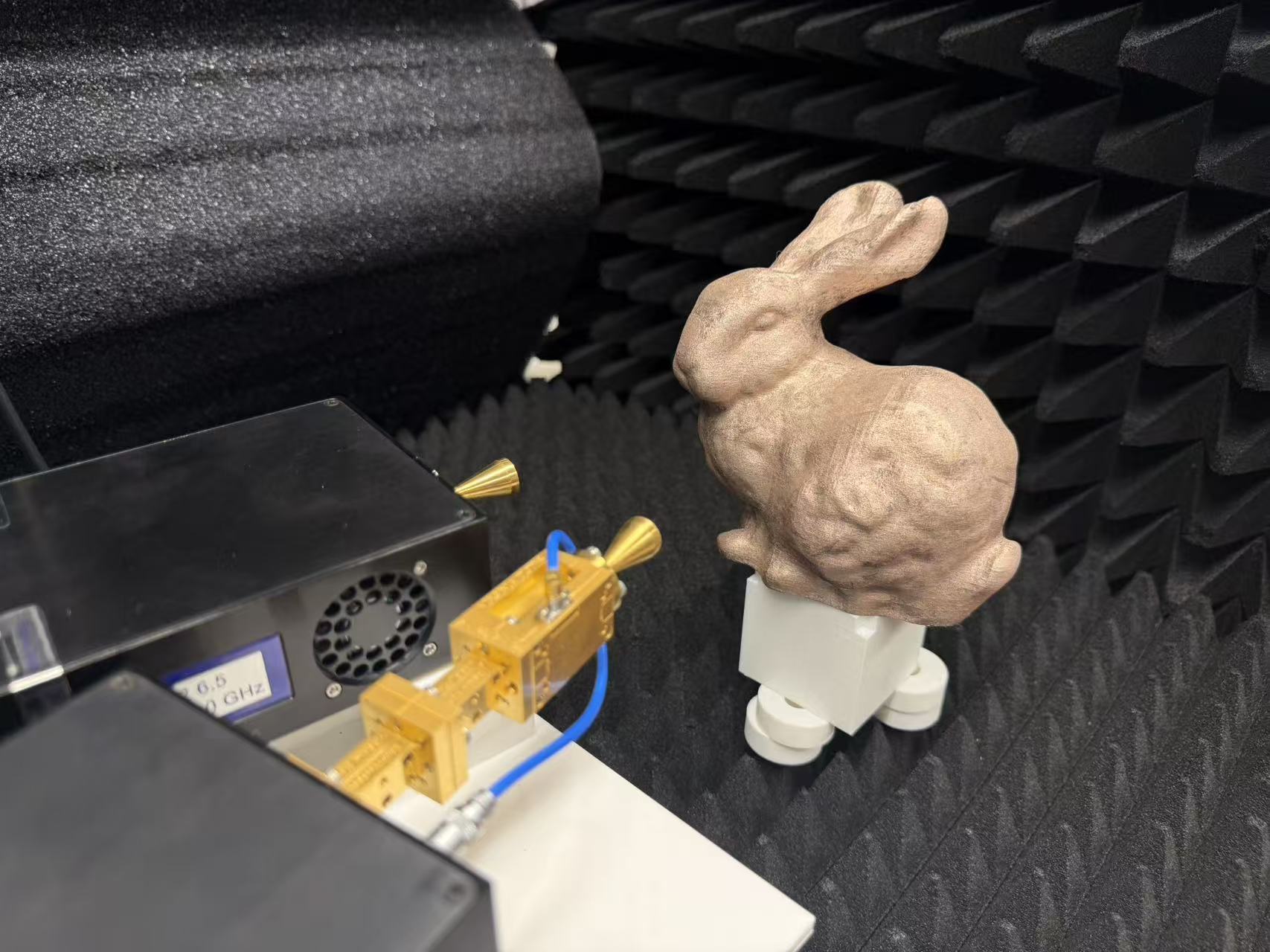}};
  \node[panel, anchor = south east] (ours) at (10, 0)
    {\includegraphics[height=.42\linewidth]{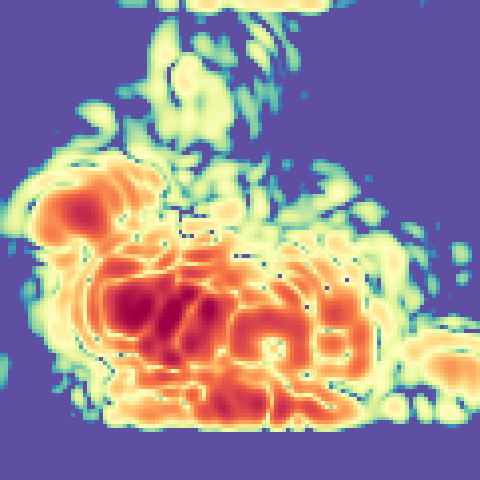}};
  \node[plab] at (setup.south west) {(a) Printed target};
  \node[plab] at (ours.south west) {(b) Ours};
\end{tikzpicture}
\captionsetup{skip=4pt}
\caption{\textbf{Reconstruction of a 3D-printed object (\S\ref{sec:results-thz}).}
\textbf{(a)}~The metallized target inside the absorber-lined enclosure.
\textbf{(b)}~Our reconstruction over the full scan window, rendered with
the measurement operator and displayed using the colormap of
Fig.~\ref{fig:thz} with gamma correction.}
\label{fig:printed-bunny}
\Description{Two panels. A photograph of the 3D-printed target, a
palm-sized figurine coated in conductive ink, standing on a narrow mount
inside absorber foam beside the two horn antennas of a sub-terahertz
transceiver, and next to it our reconstruction of that target, which shows its
outline surrounded by diffraction sidelobes.}
\vspace{-20pt}%
\end{figure}
\paragraph{3D printed object}
To evaluate the same pipeline on a real object with curvature, we
3D-print the target of \S\ref{sec:results-bunny} at \qty{120}{mm}, coat
it with conductive ink, and stand it on a narrow mount inside
an absorber-lined enclosure (Fig.~\ref{fig:printed-bunny}a).
The same transceiver scans the target confocally over a
$121{\times}121$ aperture at \qty{1.25}{mm} pitch across a
\qty{150}{mm} square window. We sweep \qty{4.8}{GHz} at a
\qty{150}{mm} focus and coherently average $32$ chirps per aperture
position. One range cell spans \qty{31.25}{mm}.
We register eight rasters and retain five after discarding three with
broadband interference. Two rasters use \qty{2.5}{mm} pitch. We also
remove aperture positions dominated by the mount return, which is
stronger than the body return in the same range cell.
The target is a continuous surface, so we fit a dense field of $4991$
Dirichlet atoms whose centers stay on a fixed lattice, with the same
kernel, forward model, and loss.
The resulting fit reaches $0.482$ NMSE against the stacked raw complex
range profiles.

Fig.~\ref{fig:printed-bunny}b shows the reconstruction. At a $35\%$
threshold over the full scan window, $93\%$ of the recovered support lies
in its largest connected component. In comparison, the measured
range-gated energy and back-projection contain $34$ and $6$ components,
respectively. Rendering the fitted field with the measurement operator
yields a dB-domain correlation of $0.74$ with the raw-scan image over the
object.
Two instrument limits bound the result: the transceiver is modeled as an
ideal confocal pair, and the aperture resolves \qty{2.5}{mm} across the
target against a range cell of \qty{31.25}{mm}.

\begingroup\ignorespaces
\subsection{Wavefront Sensing}
\label{sec:results-wavefront}

A finite pupil produces the same response in visible light, allowing the
same primitive to represent a wavefront. We simulate off-axis holograms
of four coherent beacons observed through a continuous \qty{5}{mm}
square pupil at \qty{532}{nm} under Poisson and read noise. After
sideband demodulation, each method fits the beacon positions, complex
amplitudes, and twelve shared Zernike modes to the recovered field using
the same initialization and parameter count. We generate the reference
field by Gauss--Legendre quadrature over the continuous aperture, so the
Dirichlet atoms approximate rather than generate the measurements.
Decreasing the beacon separation creates overlapping sidelobe lattices
and tests whether the methods can resolve them.
At one main-lobe FWHM of separation, the Dirichlet fit reaches $0.021$
waves of residual RMS wavefront error over fifty trials, compared with
$0.051$ for a diffraction-matched Gaussian and $0.110$ for
back-projection. Across the tested separations, it remains within two
standard errors of an oracle using the continuous-aperture response. At
separations down to half an FWHM, only the Dirichlet fit and the oracle
remain below $\lambda/30$. The Gaussian fit requires four times this
separation to meet the same threshold. At one FWHM, the Gaussian and
Dirichlet fits fail in $52\%$ and $16\%$ of trials, respectively.
\par\endgroup\ignorespaces

\section{Limitations and Future Work}
\label{sec:limitations}

\paragraph{Indirect propagation}
Multipath and inter-object scattering are modeled as
virtual surfels at the total path length, which the renderer reproduces in intensity.
The optimizer treats them as independent reflectors and
lacks the method-of-images constraint tying them to the primary
geometry, so the recovered support contains ghosts at correct path
lengths but incorrect physical locations.

\vspace{0pt plus 1fill}%
\paragraph{Occlusion}
The visibility term in Eq.~\eqref{eq:method-foreshortening} handles
back-facing surfels but not opaque self-occlusion between surfels.
In partially transmissive modalities (THz, mmWave,
synthetic-aperture sonar), occluded contributions cancel through
coherent phase rather than a discrete visibility test, and the
omission is benign for our experiments.
\vspace{0pt plus 1fill}
\paragraph{Hardware calibration}
The forward model assumes known aperture positions and a calibrated
transceiver response. The aperture extent and sampling pitch determine
the kernel. Position errors produce wavelength-scale phase errors. An
uncalibrated antenna pattern produces per-aperture amplitude errors that
the fit absorbs into $\alpha_i$. Bandwidth determines the range
resolution. Joint calibration and reconstruction is a promising
direction for future work.

\vspace{0pt plus 1fill}
\paragraph{DSFW failure modes}
DSFW has two calibration-related failure modes.
\textit{(i)}~Gross mis-scaling of $A_i$ pushes the inner ridge to absorb the
error into $|b_i|$, demoting correctly placed surfels in the utility
ranking.
\textit{(ii)}~A mis-calibrated noise floor $\sigma^2$ shifts the
certificate magnitude, causing DSFW to either fit noise as splats
(under-calibrated) or stop short of the true sparsity (over-calibrated).

\vspace{0pt plus 1fill}
\section{Conclusion}
\label{sec:conclusion}

We present \emph{Dirichlet Splatting}, a differentiable renderer for
inverse problems with wave-based coherent sensors. It replaces the
Gaussian footprint inherited from 3D Gaussian Splatting with the
analytic Dirichlet point-spread function and represents each primitive
as a surfel whose complex amplitude follows wave propagation.
Dirichlet Sliding Frank--Wolfe addresses the rugged sidelobe landscape
induced by this oscillatory kernel. We believe Dirichlet Splatting is a step
toward physically grounded differentiable rendering for radar, sonar,
and other coherent sensing modalities.

\clearpage %
\begin{acks}
We appreciate the insightful comments and feedback from the anonymous
reviewers and shepherd.  This work is partially supported by the NSF grants
2238839, 2341952, 2403124, and 2408393, and gifts from Adobe,
Google, Activision, and Meshy.
\end{acks}

\renewcommand{\bibfont}{\fontsize{7pt}{7.7pt}\selectfont}
\bibliographystyle{ACM-Reference-Format}
\bibliography{sample-base}

\appendix
\begin{table*}[t]
\centering
\small
\setlength{\tabcolsep}{6pt}
\caption{\textbf{Notation.} Symbols used throughout the paper,
grouped by role. Per-surfel quantities carry index
$i\!\in\!\{1,\dots,S\}$; $\lambda$ denotes the carrier wavelength in
physical contexts and the ridge weight in optimization contexts,
disambiguated by usage.}
\label{tab:notation}
\begin{tabularx}{\linewidth}{@{}l X@{}}
\toprule
\multicolumn{2}{@{}l}{\textit{Scene and sensor}} \\
\midrule
$\mathbf{p}_i, \mathbf{n}_i, A_i, \alpha_i, \mathbf{v}_i, \eta_i$ &
  Surfel position, normal, cell area, complex amplitude, velocity,
  wave impedance; collected in $\Theta_i$ \\
$S$, $C\!\in\!\mathbb{R}^{S\times 3}$ &
  Splat budget; stack of all surfel centers \\
$\mathbf{o}, R_i, \hat{\mathbf{r}}_i$ &
  Sensor origin; range and line-of-sight unit vector to surfel $i$ \\
$V_i, F_i, \gamma$ &
  Bistatic foreshortening, Fresnel factor, path-loss exponent \\
$F_s, K, T_c, \lambda, \delta_{\mathrm{el}}$ &
  ADC rate, chirp slope, chirp repetition interval, carrier wavelength, array
  element spacing \\
\midrule
\multicolumn{2}{@{}l}{\textit{Fourier domain and kernels}} \\
\midrule
$\ell\!\in\!\{R,D,A\}$, $N_R, N_D, N_A$ &
  Axis index and per-axis FFT lengths (range / Doppler / azimuth) \\
$\norm{\Omega}=N_R N_D N_A$ &
  Total measurement dimension \\
$\mu_{\ell,i}$ &
  Continuous Fourier coordinate of surfel $i$ on axis $\ell$
  (Eq.~\eqref{eq:method-mu}) \\%
$\Delta = k-\mu$ &
  Continuous bin offset (bin units) \\
$\kappa_N(\Delta)$, $d_N=\kappa_N/N$ &
  Complex DFT kernel ($\kappa_N(0)\!=\!N$) and its unit-peak form \\
$\kappa^w_N, \phi_N$ &
  Windowed and power (Fej\'er) kernel variants \\
\midrule
\multicolumn{2}{@{}l}{\textit{Forward model and inverse problem}} \\
\midrule
$\mathrm{splat}_i(\mathbf{f})$ &
  Per-surfel spectral splat (Eq.~\eqref{eq:method-surfel}) \\%
$\hat Y$, $\hat H=|\hat Y|^2$ &
  Rendered complex spectral cube and power image \\
$\mathbf{y}$, $H^{\mathrm{obs}}$ &
  Measured complex cube and measured power \\
$A(C)\!\in\!\mathbb{C}^{\norm{\Omega}\times S}$, $\mathbf{b}$ &
  Dictionary matrix (column $i$ = surfel footprint) and complex
  coefficient vector \\
$\mathbf{b}^\star(C), F(C)$ &
  Closed-form ridge solution and variable-projected objective
  (Eqs.~\eqref{eq:opt-bstar} and~\eqref{eq:opt-vp}) \\%
$G, q$ &
  Gram matrix $A^{\mathsf H}A/\norm{\Omega}$ and right-hand side
  $A^{\mathsf H}\mathbf{y}/\norm{\Omega}$ \\

\midrule
\multicolumn{2}{@{}l}{\textit{DSFW certificate and search}} \\
\midrule
$s(\mathbf{x})$ &
  Residual dual certificate field (Eq.~\eqref{eq:opt-cert}) \\%
$u_j, \mathrm{eff}_j\!\in\![0,1]$ &
  Per-surfel utility and coherent-cancellation efficiency
  (Eq.~\eqref{eq:opt-utility}) \\%
$r$ &
  Per-axis truncation radius for splat support (default $8$ bins) \\
\bottomrule
\end{tabularx}
\end{table*}

\section{Mathematical Derivations}
\label{supp:theory}
\label{supp:fmcw}%

\subsection{DFT response of a windowed exponential}
\label{supp:dirichlet-1d}
\label{supp:dirichlet-properties}

\begin{theorem}[DFT response of a windowed complex exponential]
\label{thm:supp-1d}
Let $x[n] = \alpha\,w[n]\exp(\j 2\pi \nu n)$ for $n=0,\dots,N{-}1$,
where $w[n]$ is a window function.
The $N$-point DFT is
\begin{equation}
X[k] = \alpha\; W_N\!\Big(\nu - \tfrac{k}{N}\Big),
\qquad
W_N(\delta) = \sum_{n=0}^{N-1} w[n]\exp(\j 2\pi \delta n).
\label{eq:supp-window-spectrum}
\end{equation}
For the rectangular window $w[n] \equiv 1$, this reduces to
\begin{equation}
X[k] = \alpha\; \e^{\j\pi(N-1)\delta}\;
\frac{\sin(\pi N \delta)}{\sin(\pi \delta)},
\qquad
\delta = \nu - \tfrac{k}{N},
\label{eq:supp-dirichlet-delta}
\end{equation}
Its magnitude is the absolute value of the Dirichlet response:
\begin{equation}
\abs{X[k]} = \abs{\alpha}\;\abs{D_N(\delta)},
\qquad
D_N(\delta) = \frac{\sin(\pi N\delta)}{\sin(\pi\delta)}.
\label{eq:supp-dirichlet}
\end{equation}
\end{theorem}

\begin{proof}
Substitution into the DFT definition gives
\[
X[k]=\alpha\sum_{n=0}^{N-1} w[n]\e^{\j2\pi(\nu-k/N)n}.
\]
For $w[n]=1$, summing the geometric series yields
\begin{equation*}
\sum_{n=0}^{N-1}\e^{\j2\pi\delta n}
=\frac{1-\e^{\j2\pi N\delta}}{1-\e^{\j2\pi\delta}}
=\e^{\j\pi(N-1)\delta}\frac{\sin(\pi N\delta)}{\sin(\pi\delta)}.
\end{equation*}
At integer $\delta$, the response follows by continuity.
\end{proof}

The bin offset used in the renderer is $\Delta=k-\mu$, with
$\mu=N\nu$, so $\delta=-\Delta/N$. Thus
\begin{equation}
X[k]=\alpha\,\e^{-\j\pi\frac{N-1}{N}\Delta}
\frac{\sin(\pi\Delta)}{\sin(\pi\Delta/N)}
=\alpha\,\kappa_N(\Delta).
\label{eq:supp-kappa-bin}
\end{equation}
The power response $|X[k]|^2=|\alpha|^2|\kappa_N(\Delta)|^2$
is proportional to the Fej\'er kernel.

\subsection{Kernel derivative}
\label{supp:gradients-kernel}

For the complex kernel
$\kappa_N(\Delta) = \e^{-\j\pi\frac{N-1}{N}\Delta}\cdot
\frac{\sin(\pi\Delta)}{\sin(\pi\Delta/N)}$,
the derivative with respect to the bin offset $\Delta$ is
\begin{equation}
\frac{\partial \kappa_N(\Delta)}{\partial \Delta}
=
\kappa_N(\Delta)\left(
-\j\pi\frac{N-1}{N}
+\pi\cot(\pi\Delta)
-\frac{\pi}{N}\cot\!\Big(\frac{\pi}{N}\Delta\Big)
\right).
\label{eq:supp-kappa-deriv}
\end{equation}

The product rule gives the phase term, and differentiating the
sine ratio gives the two cotangent terms. At integer offsets, the
factored expression must be evaluated by a limit. An equivalent
expression without singular factors follows directly from the finite sum:
\begin{equation}
\kappa_N'(\Delta)=-\frac{\j 2\pi}{N}
\sum_{n=0}^{N-1}n\,\e^{-\j 2\pi n\Delta/N}.
\label{eq:supp-kappa-deriv-sum}
\end{equation}
Appendix~\ref{supp:impl-safeguards} gives the limiting forms.

\paragraph{Geometry derivatives}
Let $a_i$ denote the amplitude factor multiplying the separable
kernel product in Eq.~\eqref{eq:method-surfel}, and let
$\Delta_{\ell,i}=k_\ell-\mu_{\ell,i}$. For a scene parameter $\xi$,
\begin{equation}
\begin{aligned}
\partial_\xi\mathrm{splat}_i
&=(\partial_\xi a_i)\prod_\ell d_{N_\ell}(\Delta_{\ell,i})\\
&\quad-a_i\sum_\ell (\partial_\xi\mu_{\ell,i})
d'_{N_\ell}(\Delta_{\ell,i})
\prod_{h\ne\ell}d_{N_h}(\Delta_{h,i}).
\end{aligned}
\end{equation}
The first term differentiates the amplitude; the second moves the
spectral footprint. Writing $\hat{\mathbf r}=\hat{\mathbf r}_i$ and
$R=R_i$, the projection in Eq.~\eqref{eq:method-mu} gives
\begin{equation}
\begin{aligned}
\nabla_{\mathbf p_i}\mu_{R,i}&=\frac{2KN_R}{cF_s}\hat{\mathbf r},\\
\nabla_{\mathbf p_i}\mu_{D,i}&=\frac{2N_DT_c}{\lambda R}
(I-\hat{\mathbf r}\hat{\mathbf r}^{\!\top})\mathbf v_i,\\
\nabla_{\mathbf v_i}\mu_{D,i}&=\frac{2N_DT_c}{\lambda}\hat{\mathbf r},\\
\nabla_{\mathbf p_i}\mu_{A,i}&=\frac{N_A\delta_{\mathrm{el}}}{\lambda R}
(I-\hat{\mathbf r}\hat{\mathbf r}^{\!\top})\mathbf u_x.
\end{aligned}
\end{equation}
The range and azimuth coordinates are independent of velocity.

\subsection{Windowed kernels}
\label{supp:dirichlet-windowed}

A taper changes the point response from $\kappa_N$ to
\begin{equation}
\kappa_N^w(\Delta)=\sum_{n=0}^{N-1}w[n]\e^{-\j 2\pi n\Delta/N}.
\end{equation}
This finite sum gives the response at fractional offsets for any
sampled window, including Hann, Hamming, and Kaiser windows.
The renderer uses the sensor's window; the main-lobe width and
sidelobe attenuation depend on that choice~\cite{harris1978use}.

\subsection{Separable spectral response}
\label{supp:dirichlet-separable}

\begin{proposition}[Separable Dirichlet structure]
\label{prop:supp-rda}
For a single point reflector that produces a separable complex
exponential of the form
$x[m,n,a] = \alpha\,\exp\!\big(\j 2\pi(\nu_D m + \nu_R n + \nu_A a)\big)$,
the 3D DFT (with optional separable windows
$w_D[m], w_R[n], w_A[a]$) factorizes as
\begin{equation}
Y[p,k,q]
=
\alpha\;
W_M\!\Big(\nu_D - \tfrac{p}{M}\Big)\;
W_N\!\Big(\nu_R - \tfrac{k}{N}\Big)\;
W_A\!\Big(\nu_A - \tfrac{q}{A}\Big),
\label{eq:supp-3d-factor}
\end{equation}
where each $W$ is a 1D window spectrum as in
Eq.~\eqref{eq:supp-window-spectrum}.
For rectangular windows, each factor is (up to phase) a Dirichlet
kernel, and the point response is therefore a separable product
of Dirichlet kernels across the three Fourier axes.
\end{proposition}

\begin{proof}
The signal and window separate over $m$, $n$, and $a$, so the
windowed DFT factorizes into three sums of the form
Eq.~\eqref{eq:supp-window-spectrum}.
\end{proof}

We evaluate the point response as a product of three 1D kernels.

\subsection{FMCW separable measurement model}
\label{supp:fmcw-chirp}

For FMCW acquisition~\cite{richards2014fundamentals},
a baseband chirp of duration $T$, slope $K=B/T$, and carrier $f_c$,
dechirped against the transmit reference under the stop-and-hop
approximation, returns a complex exponential at the beat frequency
$f_b = 2KR/c$ for a reflector at range $R$.
Across $M$ chirps with repetition $T_c$, target motion at radial
velocity $v_r$ induces a slow-time Doppler $f_D = 2v_r/\lambda$ with
$\lambda = c/f_c$.
Using the azimuth sign convention of Eq.~\eqref{eq:method-mu},
an $A$-element uniform linear array with spacing $d$ has a far-field steering
phase across antennas is $\exp(\j 2\pi a\,u(\theta))$ with
$u(\theta)=(d/\lambda)\sin\theta$.
The sampled dechirped signal across (slow-time, fast-time, array) is
therefore separable:
\begin{equation}
x[m,n,a]
=
\alpha\,
\exp\!\Big(\j 2\pi\big(f_D\,mT_c + f_b\,nT_s + a\,u(\theta)\big)\Big),
\label{eq:supp-samples-rd}
\end{equation}
which matches the canonical form of Proposition~\ref{prop:supp-rda}
with $\nu_D=f_D T_c$, $\nu_R=f_b T_s$, $\nu_A=u(\theta)$.
The Fourier-processing pipeline therefore reduces to evaluating one
windowed DFT per axis, whose closed form is the Dirichlet kernel of
Eqs.~\eqref{eq:supp-dirichlet}--\eqref{eq:supp-kappa-bin}.

\subsection{Validity of the separable model}
\label{supp:fmcw-validity}

Equation~\eqref{eq:supp-samples-rd} assumes stop-and-hop acquisition,
constant radial velocity, and negligible range walk under the far-field
array model above. We also give the interval in which velocity can be
recovered without aliasing.

\paragraph{(A1) Stop-and-hop}
The target is treated as stationary during a single chirp
(fast-time), with motion accumulating only between chirps
(slow-time).
This drops a within-chirp contribution
$\exp(\j 2\pi f_D\, n T_s)$ that is negligible whenever
\begin{equation}
|v_r| \;\ll\; \frac{c}{2 f_c\, N_R\, T_s}.
\label{eq:supp-stophop-bound}
\end{equation}
\paragraph{(A2) Constant radial velocity}
The slow-time phase is linear only for constant radial velocity
over the coherent processing interval $T_{\mathrm{CPI}}=N_DT_c$.
Radial acceleration $a_r$ adds a quadratic phase proportional to
$a_r(mT_c)^2/\lambda$, producing a Fresnel-type response.
This contribution is negligible when
\begin{equation}
|a_r|\, T_{\mathrm{CPI}}^2 \;\ll\; \lambda.
\label{eq:supp-accel-bound}
\end{equation}
Acceleration outside this regime requires a different Doppler response.

\paragraph{(A3) Negligible range walk}
A target with non-zero radial velocity traverses
$\Delta R_{\mathrm{walk}} = v_r\, T_{\mathrm{CPI}}$ during the CPI,
which couples range and Doppler once $\Delta R_{\mathrm{walk}}$
exceeds one range bin
$\Delta R_{\mathrm{bin}} = c/(2B)$ (sweep bandwidth $B$).
Separability of $\prod_\ell d_{N_\ell}(\cdot)$ in
Eq.~\eqref{eq:method-surfel} therefore requires
\begin{equation}
|v_r|\, T_{\mathrm{CPI}} \;\ll\; \frac{c}{2 B}.
\label{eq:supp-walk-bound}
\end{equation}
\paragraph{(A4) Doppler unambiguity}
The slow-time sampling imposes the Nyquist bound $|f_D| < 1/(2 T_c)$. Radial
velocities are therefore unambiguous only on
\begin{equation}
|v_r| \;<\; \frac{\lambda}{4 T_c}.
\label{eq:supp-amb-bound}
\end{equation}
Velocities separated by integer multiples of $\lambda/(2T_c)$
produce the same sampled Doppler response. This ambiguity affects
velocity recovery, not the separability of the forward response.

\begingroup\ignorespaces
\subsection{Fresnel response at a planar boundary}
\label{supp:fresnel}

\paragraph{Setup}
A surfel is a locally planar boundary between the surrounding medium,
of wave impedance $\eta_0 = \sqrt{\mu_0/\varepsilon_0}$, and the surfel
material, of wave impedance $\eta_i = \sqrt{\mu_i/\varepsilon_i}$.
A plane wave arrives at angle $\theta_{\mathrm{inc}}$ measured from the
surfel normal $\mathbf{n}_i$ and excites a transmitted wave at angle
$\theta_t$. Matching the tangential electric and magnetic fields across
the boundary gives the classical two-equation system whose solution is
the reflection coefficient below.

\paragraph{Transverse-electric coefficient}
When the incident electric field is perpendicular to the plane of
incidence, tangential-field matching yields
\begin{equation}
F_{\mathrm{TE}}
=
\frac{\eta_i\cos\theta_{\mathrm{inc}} - \eta_0\cos\theta_t}
     {\eta_i\cos\theta_{\mathrm{inc}} + \eta_0\cos\theta_t},
\label{eq:supp-fresnel-te}
\end{equation}
which is Eq.~\eqref{eq:method-fresnel}.
The transmitted angle follows from Snell's law,
$\varepsilon_0^{1/2}\sin\theta_{\mathrm{inc}}
 = \varepsilon_i^{1/2}\sin\theta_t$, which for non-magnetic media
($\mu_i = \mu_0$) simplifies to
$\sin\theta_t/\sin\theta_{\mathrm{inc}} = \eta_i/\eta_0$.

\paragraph{Transverse-magnetic coefficient}
When the incident magnetic field is perpendicular to the plane of
incidence, the same matching argument with the roles of the two fields
exchanged gives
\begin{equation}
F_{\mathrm{TM}}
=
\frac{\eta_i\cos\theta_t - \eta_0\cos\theta_{\mathrm{inc}}}
     {\eta_i\cos\theta_t + \eta_0\cos\theta_{\mathrm{inc}}}.
\label{eq:supp-fresnel-tm}
\end{equation}
Our sub-THz transceiver is linearly polarized and the targets are
oriented so that the transverse-electric branch dominates, so the
renderer evaluates Eq.~\eqref{eq:supp-fresnel-te}.

\paragraph{Two limits}
For a perfect conductor, $\abs{\eta_i}\to 0$ and
Eq.~\eqref{eq:supp-fresnel-te} gives $F\to-1$ at every incidence angle,
so the angular dependence collapses and $\alpha_i$ alone determines the
splat amplitude. In the dielectric case, where $\eta_i$ is comparable to
$\eta_0$, $\abs{F}$ varies over the aperture by a factor that grows with
the angular extent of the scan, and fitting $\alpha_i$ alone cannot
absorb that variation because $\alpha_i$ is constant across the aperture.

\paragraph{Identifiability}
The measurement constrains only the product $\alpha_i F_i$.
Separating the two parameters requires an aperture spanning enough
incidence angles to reveal the angular variation of $F$.
Under a narrow aperture the pair is degenerate up to a scale, and we
hold $\eta_i$ at its material prior. The reported sub-THz experiments use metal
targets in the $\abs{\eta_i}\to 0$ regime, so impedance is not estimated.
\par\endgroup\ignorespaces

\section{Implementation Details}
\label{supp:impl}
\label{app:optimizer}%
\label{supp:fmcw-device-params}%

\subsection{Sensor parameters}
\label{supp:impl-device}

The projection in Eq.~\eqref{eq:method-mu} uses sensor constants
read from the device datasheet or programmed waveform profile.
Table~\ref{tab:fmcw-device-params} maps these constants to the two
front-ends used in this paper. None are learned.

\begin{table*}[t]
\centering
\caption{\textbf{Mapping from device configuration to the constants
in Eq.~\eqref{eq:method-mu}.}
The mmWave column corresponds to the FMCW Doppler-validation experiment
of Appendix~\ref{app:doppler-validation}. The sub-THz column corresponds
to the THz application of \S\ref{sec:results-thz}.
}
\label{tab:fmcw-device-params}
\small
\setlength{\tabcolsep}{4pt}
\begin{tabular}{lll}
\toprule
\textbf{Constant} & \textbf{mmWave FMCW radar} & \textbf{Sub-THz SDR} \\
\midrule
$f_c$ (carrier frequency)       & $77$~GHz (waveform-profile $f_{\mathrm{start}}$) & $120$~GHz (LO synthesizer setpoint) \\
$\lambda = c/f_c$               & $3.90$~mm                                       & $2.50$~mm \\
$B$ (sweep bandwidth)           & $250$~MHz (frequency-slope $\times$ $T$)        & $1.23$~GHz (programmed sweep span) \\
$T$ (chirp duration)            & user-set fast-time window                       & $8192\,T_s$ (acquisition window) \\
$K = B/T$ (chirp slope)         & derived; sets range-bin scale via $2K/c$        & derived from $B$ and $T$ \\
$T_c$ (chirp repetition interval) & $60~\mu$s (chirp-cycle-time register)         & N/A (single-shot acquisition) \\
$F_s$ (ADC sampling rate)       & ADC sampling-rate register                      & SDR-card sample-rate setting \\
$N_R$                           & ADC samples per chirp                           & $8192$ (fast-time samples) \\
$N_D$                           & chirps per frame                                & captures per frame \\
$N_A$                           & $N_{\mathrm{tx}} N_{\mathrm{rx}}$ (virtual ULA) & raster-scan grid size \\
$\delta_{\mathrm{el}}$ (array element spacing) & $\lambda/2$ (cascade EVM)        & raster step (mechanical) \\
\bottomrule
\end{tabular}
\end{table*}

\paragraph{Bin-scale consistency}
Substituting the FMCW identities
\[
F_s=KT\cdot(N_R/N_R\,\text{samples per chirp}),\qquad
\Delta R_{\mathrm{bin}}=c/(2B)
\]
reduces $\partial\mu_R/\partial R=(2K/F_s)(N_R/c)$ to
$N_R/\Delta R_{\mathrm{bin}}$, giving a unit bin shift for one
range-bin displacement. The Doppler scale $2N_DT_c/\lambda$ maps
the velocity-bin width $\lambda/(2T_cN_D)$ to one bin; the azimuth
scale $N_A\delta_{\mathrm{el}}/\lambda$ does the same on the spatial
axis. Porting the renderer therefore requires only the device
constants $(2K/cF_s,2T_c/\lambda,\delta_{\mathrm{el}}/\lambda)$.
The remaining forward model is shared across modalities.

\subsection{Renderer truncation and noise floor}
\label{supp:impl-renderer}

\paragraph{Truncation}
We truncate each splat's contribution to a radius $r=8$ bins per
axis, which reduces the per-step cost from
$O(S\,\prod_{\ell=1}^{L} N_\ell)$ to $O(S\,r^L)$ and makes the
renderer scale linearly in surfel count rather than in
spectral-tensor volume.
For the rectangular window with $N\in[64,256]$, $r=8$ captures
$>99\%$ of the kernel energy. Tapered windows (Hann, Hamming) are
even more concentrated, and the same radius remains safe.

\paragraph{Noise floor}
We add a learned or calibrated Gaussian noise floor $\sigma^2$ to
the power image: $\hat H=|\hat Y|^2+\sigma^2$.
Section~\ref{sec:opt} describes the loss design, including log-power
supervision for the 40--60~dB dynamic range of coherent measurements.

\paragraph{Matrix-free accumulation}
Materializing $A(C)\in\mathbb{C}^{|\Omega|\times S}$ for
Eq.~\eqref{eq:opt-bstar} exceeds $1$~GiB in single-precision complex
for our 3D settings. We instead stream voxel chunks and accumulate
$G=A^{\!\mathsf H}A$ and $q=A^{\!\mathsf H}\mathbf y$.
Since $G$ is Hermitian, we store its upper triangle, reducing the
working set to $O(\mathrm{chunk}\cdot S)$.
A fused CUDA kernel combines footprint evaluation, conjugate
multiplication, and reduction without intermediate writes.
The streamed path is numerically equivalent to the dense reference
and reduces peak memory and per-step time by over an order of
magnitude (Appendix~\ref{app:runtime}).

\subsection{Training stabilizers and numerical safeguards}
\label{supp:impl-safeguards}

\paragraph{Periodic coupled correction}
To account for the off-diagonal coherent coupling in \S\ref{sec:opt},
every $T_{\mathrm{LM}}\!=\!40$ outer steps we therefore take a global
Gauss--Newton step on a downsampled spectral target, using the
VarPro-reduced center Jacobian
$J_{\mathrm{red}} = (I - P_A) J_c$~\cite{golub1973variable} with a
trust-radius clip and a full-resolution acceptance check.
The periodic correction implements the low-resolution
acceptance sub-step of Algorithm~\ref{alg:dsfw}.

\paragraph{Window matching}
The analytic kernel $\kappa^w_N$ uses the same taper as the sensor,
so that the rendered point response matches the deployed measurement
chain to within numerical precision.

\paragraph{Kaiser-to-target kernel-continuation curriculum}
Training begins with a smooth Kaiser kernel and anneals the shape
parameter toward the deployed taper, mirroring the
multi-resolution spectral-estimation practice of windowed-FFT
super-resolution~\cite{harris1978use}.
This restores sidelobe contrast while avoiding early sidelobe peaks.

\paragraph{Pairwise repulsion penalty}
A short-range inter-splat penalty prevents centers from collapsing
under coherent cancellation.

\paragraph{Coarse-to-fine schedule on kernel order $N$}
Short windows are used early (broad mainlobes, weak sidelobes) and
the full $N$ near convergence, so that initial certificate peaks are
well separated and final accuracy is set by the deployed kernel.

\paragraph{Kernel evaluation near integer offsets}
The derivative formula
$\partial \kappa_N/\partial \Delta$ of
Appendix~\ref{supp:gradients-kernel} contains two cotangent terms
that individually diverge at integer bin offsets but whose
combination remains finite.
Near $\Delta\approx 0$,
\[
\lim_{\Delta\to 0}\left[\pi\cot(\pi\Delta)-\frac{\pi}{N}\cot\!\Big(\frac{\pi}{N}\Delta\Big)\right]
= 0,
\]
We switch to a Taylor branch when
$|\Delta|<\epsilon$, with $\epsilon=10^{-4}$ throughout:
\[
\begin{aligned}
\pi\cot(\pi\Delta)&\approx\tfrac1\Delta-\tfrac{\pi^2\Delta}{3}+O(\Delta^3),\\
\tfrac\pi N\cot\!\Big(\tfrac{\pi\Delta}{N}\Big)
&\approx\tfrac1\Delta-\tfrac{\pi^2\Delta}{3N^2}+O(\Delta^3).
\end{aligned}
\]
The difference $-\tfrac{\pi^2\Delta}{3}(1-1/N^2) + O(\Delta^3)$ is
well-behaved at the origin.
The same $\epsilon$-regularization is applied near every $\Delta$
that is a non-zero multiple of $N$, where $\sin(\pi\Delta/N)=0$.
This branch switch is the only floating-point safeguard in the
forward and backward kernel evaluations.

\subsection{Hyperparameters and Reproducibility}
\label{supp:impl-hyperparams}

Table~\ref{tab:supp-hyperparams} lists the experiment settings,
fixed across seeds. Only the seed and initial centers $C_0$ vary.

\begin{table}[tbp]
\centering
\small
\caption{\textbf{DSFW hyperparameters per experiment.}
$T$: outer iterations; $\lambda$: VarPro ridge weight;
$r$: truncation radius in bins; $\mathcal{T}_{\mathrm{lr}}$:
low-resolution sub-step interval; $\beta_0$: initial Kaiser shape.}
\label{tab:supp-hyperparams}
\setlength{\tabcolsep}{4pt}
\begin{tabular}{lccccc}
\toprule
Experiment & $T$ & $\lambda$ & $r$ & $\mathcal{T}_{\mathrm{lr}}$ & $\beta_0$ \\
\midrule
2D synthetic ($128{\times}128$, $S=50$) & $200$ & $10^{-4}$ & $8$ & every $20$ & $14$ \\
3D synthetic ($64{\times}256{\times}256$) & $300$ & $10^{-4}$ & $4$ & every $25$ & $14$ \\
THz planar target                         & $400$ & $10^{-3}$ & $8$ & every $20$ & $10$ \\
mmWave automotive cube                    & $200$ & $10^{-4}$ & $4$ & every $25$ & $14$ \\
\bottomrule
\end{tabular}
\end{table}

\begingroup\ignorespaces
\paragraph{Baseline optimizer settings}
The first-order controls run for $300$ steps under the paper preset.
Dirichlet-AdamW and SGD
use center/coefficient learning rates $8\times10^{-3}$/$2.5\times10^{-2}$,
with AdamW weight decay $10^{-4}$ and SGD momentum $0.9$.
The 3DGS-style, GSRF, and WRF-GS+ baselines use AdamW with learning
rate $0.03$ and weight decay $10^{-4}$.
The quasi-Newton control follows $300$ AdamW steps with $50$ L-BFGS
iterations (learning rate $0.8$, history $20$, strong-Wolfe line search).
AdamW uses the PyTorch default moment coefficients $(0.9,0.999)$.
\par\endgroup\ignorespaces

All experiments run on a single NVIDIA RTX A6000 with PyTorch
$2.1$ and CUDA $12.1$.
Metrics average $10$ seeds, controlling the uniform initial centers
$C_0$ and tie-breaking among low-utility candidates.

\begingroup\ignorespaces
\paragraph{Candidate volume and initialization}
Table~\ref{tab:supp-init} lists the candidate grids and splat budgets.
The candidate grid spans the sensor's range gate and azimuth field
of view. Each comparison shares this grid and budget $S$, with
cold-start centers drawn uniformly without replacement.
The multi-view fit of \S\ref{sec:results-bunny} and the printed-object fit of \S\ref{sec:results-thz} seed their centers from back-projection peaks instead.

\begin{table}[tbp]
\centering
\small
\caption{\textbf{Candidate grid, range gate, and splat budget.}
Each comparison shares budget $S$. The raw-signal budget is capped
by both the listed cap and the scene reflector count.
The sub-THz candidate pitch is $0.75\times$ the beam sigma,
upsampled $2\times$.}
\label{tab:supp-init}
\setlength{\tabcolsep}{4pt}
\begin{tabular}{lccc}
\toprule
Experiment & Candidate grid & Range gate & $S$ \\
\midrule
Signal-space fitting (\S\ref{sec:results-inverse}) & $128{\times}128$   & full axis & $50$ \\
Basin sweep (\S\ref{sec:results-inverse})          & $128{\times}128$   & full axis & $5$ \\
Raw-signal imaging (\S\ref{sec:results-3d})   & $56{\times}56$     & full axis & ${\leq}192$ \\
Sub-THz letters (\S\ref{sec:results-thz})      & beam-derived pitch & ROI gate  & $512$ \\
\bottomrule
\end{tabular}
\end{table}
\par\endgroup\ignorespaces

\section{Additional experiments}
\label{app:extra}
\setlength{\dblfloatsep}{4pt plus 1pt minus 1pt}
\setlength{\dbltextfloatsep}{6pt plus 1pt minus 1pt}

\begin{figure*}[!t]
\begin{minipage}[t]{0.48\textwidth}
\vspace{0pt}
\centering
\includegraphics[width=\linewidth]{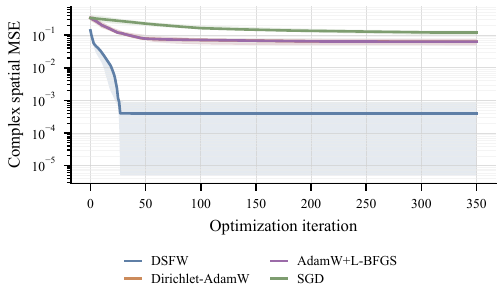}
\caption{\textbf{Cross-scene convergence trajectories.}
Median complex spatial MSE with interquartile shading over the seven
signal-space scenes in Fig.~\ref{fig:inverse-result-grid}.}
\Description{Median loss versus iteration for four optimizers, with interquartile bands over seven scenes.}
\label{fig:supp-convergence-suite}
\end{minipage}
\hfill
\begin{minipage}[t]{0.48\textwidth}
\vspace{0pt}
\centering
\includegraphics[width=\linewidth]{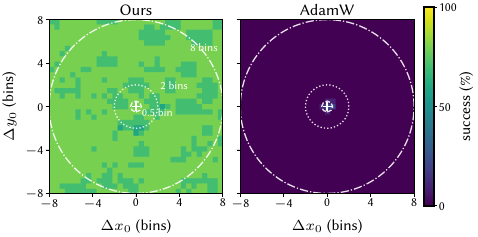}
\caption{%
  \textbf{Basin of attraction.}
  Two-dimensional cold-start success over initial position offset on
  $K=5$ reflector scenes, averaged over ten random seeds.
  The certificate update gives DSFW a broad attraction basin across the
  tested eight-bin range ($77.4\%$ aggregate success), while
  Dirichlet-AdamW succeeds on only $0.07\%$ of trials.}
\Description{Two heatmaps of cold-start recovery success versus horizontal and vertical initial offset, comparing DSFW and Dirichlet-AdamW.}
\label{fig:basin}
\end{minipage}
\end{figure*}

\subsection{Doppler validation}
\label{app:doppler-validation}

\begingroup
\emergencystretch=1em
We validate the Doppler parameter under the constant-\allowbreak radial-\allowbreak velocity,
no-range-walk assumptions of Appendix~\ref{supp:fmcw-validity}.
Acceleration, tracking, occlusion changes, and dynamic topology
are outside this validation's scope.
\par\endgroup

\paragraph{Protocol}
We synthesize a three-reflector range--Doppler--azimuth cube with
$N_R=96$, $N_D=64$, and $N_A=8$.
The FMCW constants are carrier frequency $f_c=77$\,GHz, chirp
interval $T_c=60$\,\textmu s, and bandwidth $B=250$\,MHz, giving a
velocity-bin spacing of $0.507$\,m/s.
The three targets have radial velocities $-6.0$, $0.0$, and
$7.5$\,m/s, with maximum range walk of $0.051$ range bins across the
coherent processing interval.
We compare three rows:
\textit{(i)} the finite-window FFT synthesis oracle;
\textit{(ii)} a joint Dirichlet position-and-velocity model that searches the
radial velocity coordinate on a $0.05$\,m/s grid over $\pm 8.0$\,m/s,
$321$ candidates, and refits complex amplitudes by least squares; and
\textit{(iii)} a static position-only Dirichlet control with the same
range and azimuth support but Doppler center fixed to zero.

\begin{table}[tbp]
\centering
\caption{\textbf{Doppler validation metrics.}
The joint velocity model recovers the moving-reflector cube, while
the static control cannot explain Doppler-shifted energy.}
\label{tab:doppler-validation}
\small
\setlength{\tabcolsep}{3pt}
\begin{tabular}{lcc}
\toprule
\textbf{Method} & Velocity RMSE (m/s)$\downarrow$ & Complex NMSE (dB)$\downarrow$ \\
\midrule
FFT reference & $0.000$ & $-120.00$ \\
Joint Dirichlet velocity & $0.000$ & $-120.00$ \\
Static position-only & $5.545$ & $-1.11$ \\
\bottomrule
\end{tabular}
\end{table}

The joint velocity model matches the FFT oracle to the numerical
floor, while clamping velocity to zero on the same scene degrades
the fit by more than two orders of magnitude in complex NMSE.
This isolates the contribution of the Doppler parameter itself,
holding range and angle fixed.

\paragraph{Randomized sweep}
We repeat the Doppler validation on eight randomized three-reflector
scenes with off-bin moving-target velocities
($1.25{\leq}|v_r|{\leq}8.0$~m/s) and complex Gaussian noise at
40, 30, 20, and 10~dB SNR. We use the same range and azimuth support
with a coarser 0.25~m/s velocity-search step. The joint model obtains
0.074~m/s mean velocity RMSE at all four SNRs. Its mean complex NMSE
changes from $-11.97$~dB at 40~dB SNR to $-8.14$~dB at 10~dB SNR.
The static control remains near $0$~dB, giving static-minus-joint NMSE
gaps of $11.97$~dB and $8.13$~dB at those two endpoints. This experiment
tests the Doppler parameter across randomized velocities and noise
levels. It does not evaluate dynamic-scene reconstruction.

\begin{figure*}[!t]
\begin{minipage}[t]{0.48\textwidth}
\vspace{0pt}
\centering
\includegraphics[width=\linewidth]{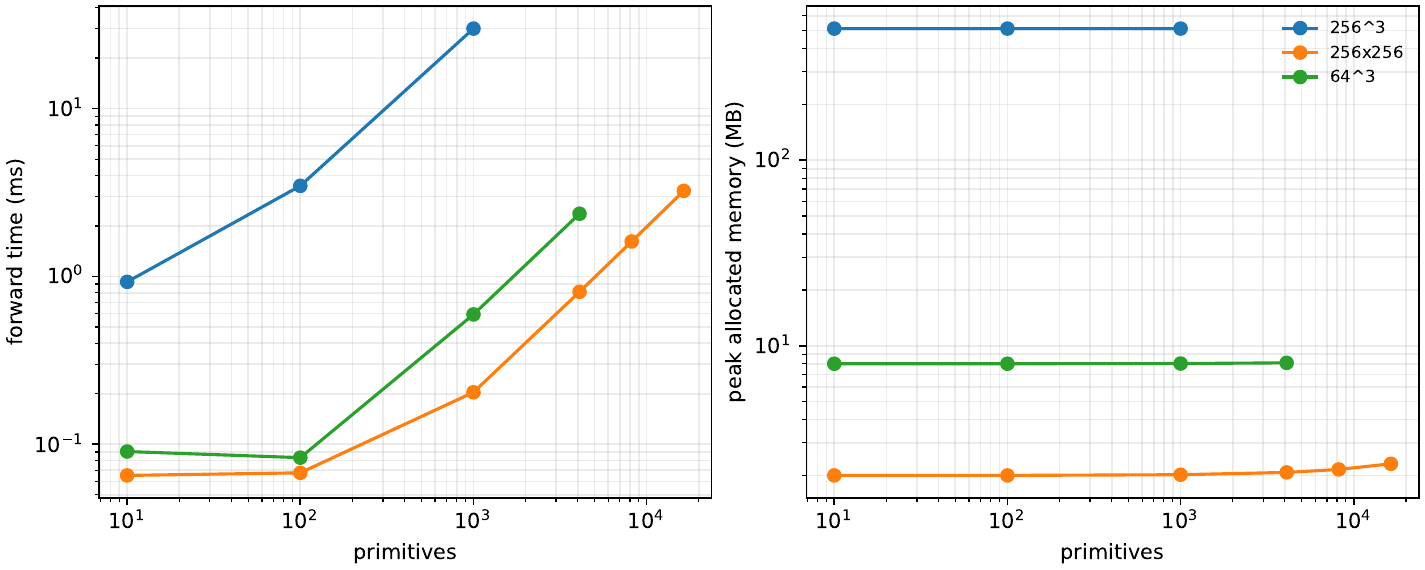}
\caption{\textbf{Forward render runtime and memory scaling.}
Forward render time scales with primitive count, while peak allocation
is dominated by the output grid size for these fused-renderer settings.}
\Description{Two log-log plots showing forward render time and peak allocated memory versus primitive count for 2D and 3D grids.}
\label{fig:runtime-slang-scaling}
\end{minipage}
\hfill
\begin{minipage}[t]{0.48\textwidth}
\vspace{0pt}
\centering
\captionof{table}{\textbf{Forward runtime and memory.}
Steady-state fused Dirichlet forward render time and peak GPU
allocation.}
\label{tab:runtime-slang}
\small
\setlength{\tabcolsep}{4pt}
\begin{tabular}{lrrr}
\toprule
\textbf{Grid} & \textbf{Primitives} & \textbf{Time (ms)} & \textbf{Peak MB} \\
\midrule
$256{\times}256$ & 10 & 0.065 & 2.0 \\
$256{\times}256$ & 100 & 0.068 & 2.0 \\
$256{\times}256$ & 1000 & 0.204 & 2.0 \\
$256{\times}256$ & 4096 & 0.810 & 2.1 \\
$256{\times}256$ & 8192 & 1.615 & 2.2 \\
$256{\times}256$ & 16384 & 3.236 & 2.3 \\
$64^3$ & 10 & 0.090 & 8.0 \\
$64^3$ & 100 & 0.083 & 8.0 \\
$64^3$ & 1000 & 0.594 & 8.0 \\
$64^3$ & 4096 & 2.363 & 8.1 \\
$256^3$ & 10 & 0.928 & 512.0 \\
$256^3$ & 100 & 3.466 & 512.0 \\
$256^3$ & 1000 & 29.984 & 512.0 \\
\bottomrule
\end{tabular}
\end{minipage}
\end{figure*}

\subsection{Cross-scene quantitative summaries}
\label{app:cross-scene-statistics}

Tables report mean $\pm$ standard deviation across the displayed
scenes in the signal-space fitting and raw-signal imaging experiments.
PSNR and SSIM are computed on the displayed reconstruction panels.

\begin{table*}[t]
\centering
\caption{\textbf{Cross-scene inverse signal-space statistics.}
Mean $\pm$ standard deviation over the seven signal-space scenes in Fig.~\ref{fig:inverse-result-grid}.
PSNR and SSIM are computed on the displayed magnitude fields against the reference target.}
\label{tab:supp-inverse-signal-stats}
\small
\setlength{\tabcolsep}{3.5pt}
\begin{tabular}{lcccc}
\toprule
\textbf{Method} & \textbf{Complex MSE}$\downarrow$ & \textbf{Mag. NMSE (dB)}$\downarrow$ & \textbf{PSNR}$\uparrow$ & \textbf{SSIM}$\uparrow$ \\
\midrule
3DGS & $1.92{\times}10^{-1} \pm 1.18{\times}10^{-1}$ & $-3.09 \pm 0.59$ & $17.44 \pm 4.94$ & $0.396 \pm 0.154$ \\
GSRF & $4.43{\times}10^{-2} \pm 3.02{\times}10^{-2}$ & $-10.90 \pm 3.75$ & $25.26 \pm 3.71$ & $0.654 \pm 0.135$ \\
WRF-GS+ & $4.38{\times}10^{-2} \pm 3.04{\times}10^{-2}$ & $-11.01 \pm 3.74$ & $25.37 \pm 3.68$ & $0.650 \pm 0.126$ \\
SGD & $1.12{\times}10^{-1} \pm 5.57{\times}10^{-2}$ & $-3.98 \pm 2.51$ & $18.34 \pm 3.12$ & $0.542 \pm 0.108$ \\
Dirichlet-AdamW & $6.06{\times}10^{-2} \pm 3.32{\times}10^{-2}$ & $-6.62 \pm 3.40$ & $20.98 \pm 3.36$ & $0.712 \pm 0.128$ \\
AdamW+L-BFGS & $5.97{\times}10^{-2} \pm 3.27{\times}10^{-2}$ & $-6.69 \pm 3.44$ & $21.04 \pm 3.35$ & $0.717 \pm 0.131$ \\
Ours & {\boldmath\textbf{$6.71{\times}10^{-4} \pm 9.06{\times}10^{-4}$}} & {\boldmath\textbf{$-38.54 \pm 18.53$}} & {\boldmath\textbf{$52.89 \pm 18.69$}} & {\boldmath\textbf{$0.998 \pm 0.003$}} \\
\bottomrule
\end{tabular}
\end{table*}

\paragraph{Convergence suite}
To complement the representative loss trace in the main paper,
Fig.~\ref{fig:supp-convergence-suite} aggregates the recorded histories
over the seven scenes in Fig.~\ref{fig:inverse-result-grid}.

\begin{table*}[t]
\centering
\caption{\textbf{Cross-scene convergence statistics.}
Mean $\pm$ standard deviation over the seven signal-space scenes in Fig.~\ref{fig:inverse-result-grid}.
$\rho$ is the empirical exponential loss-decay rate fitted to each recorded trace.
The threshold column reports the mean iteration among scenes that reach complex MSE $10^{-3}$, followed by the reach count.}
\label{tab:supp-convergence-suite}
\small
\setlength{\tabcolsep}{3.5pt}
\begin{tabular}{lccc}
\toprule
\textbf{Method} & \textbf{$\rho$}$\uparrow$ & \textbf{Iter. to $10^{-3}$}$\downarrow$ & \textbf{Final MSE}$\downarrow$ \\
\midrule
DSFW & {\boldmath\textbf{$0.790 \pm 1.525$}} & $18.6 \pm 11.1$ (5/7) & {\boldmath\textbf{$6.71{\times}10^{-4} \pm 9.06{\times}10^{-4}$}} \\
Dirichlet-AdamW & $0.064 \pm 0.010$ & -- (0/7) & $6.06{\times}10^{-2} \pm 3.32{\times}10^{-2}$ \\
AdamW+L-BFGS & $0.063 \pm 0.010$ & -- (0/7) & $5.97{\times}10^{-2} \pm 3.27{\times}10^{-2}$ \\
SGD & $0.012 \pm 0.005$ & -- (0/7) & $1.12{\times}10^{-1} \pm 5.57{\times}10^{-2}$ \\
\bottomrule
\end{tabular}
\end{table*}

\paragraph{Basin of attraction}
Fig.~\ref{fig:basin} visualizes the two-dimensional cold-start
success map summarized in \S\ref{sec:results-inverse}.

\begin{table*}[t]
\centering
\caption{\textbf{Cross-scene raw-signal inverse-imaging statistics.}
Mean $\pm$ standard deviation over the seven displayed planar targets in Fig.~\ref{fig:known-plane-reconstruction}.
PSNR and SSIM are computed on object-plane reconstruction panels against the reference support.}
\label{tab:supp-plane-reconstruction-stats}
\small
\setlength{\tabcolsep}{3.5pt}
\begin{tabular}{lccccc}
\toprule
\textbf{Method} & \textbf{Support NMSE}$\downarrow$ & \textbf{PSNR}$\uparrow$ & \textbf{SSIM}$\uparrow$ & \textbf{IoU}$\uparrow$ & \textbf{Dice}$\uparrow$ \\
\midrule
Back-projection & $0.387 \pm 0.252$ & $13.25 \pm 2.02$ & $0.601 \pm 0.106$ & $0.612 \pm 0.211$ & $0.740 \pm 0.176$ \\
3DGS & $0.389 \pm 0.266$ & $13.23 \pm 1.69$ & $0.392 \pm 0.088$ & $0.628 \pm 0.207$ & $0.753 \pm 0.171$ \\
GSRF & $0.406 \pm 0.248$ & $12.98 \pm 1.82$ & $0.554 \pm 0.045$ & $0.590 \pm 0.204$ & $0.723 \pm 0.171$ \\
WRF-GS+ & $0.391 \pm 0.250$ & $13.19 \pm 1.94$ & $0.602 \pm 0.048$ & $0.610 \pm 0.208$ & $0.739 \pm 0.173$ \\
GS+SFW & $0.294 \pm 0.306$ & $15.38 \pm 3.00$ & $0.683 \pm 0.130$ & $0.740 \pm 0.269$ & $0.823 \pm 0.208$ \\
Ours & {\boldmath\textbf{$0.136 \pm 0.109$}} & {\boldmath\textbf{$19.41 \pm 5.51$}} & {\boldmath\textbf{$0.810 \pm 0.061$}} & {\boldmath\textbf{$0.801 \pm 0.150$}} & {\boldmath\textbf{$0.883 \pm 0.095$}} \\
\bottomrule
\end{tabular}
\end{table*}

\subsection{Per-component ablations}
\label{app:per-component-ablations}

\S\ref{sec:results-ablations} reports the $2{\times}2$ kernel--optimizer
factorial.
Tab.~\ref{tab:e4-factorial-rerun} reports a rerun under the paper
preset ($10$ seeds, $128{\times}128$ grid, $5$ splats), which adds the
center RMSE column.

\begin{table}[H]
\centering
\caption{\textbf{Kernel-and-optimizer factorial rerun.}
The Dirichlet kernel and SFW/VarPro optimizer are both required for
reliable cold-start recovery under the E4 paper preset.}
\label{tab:e4-factorial-rerun}
\small
\setlength{\tabcolsep}{4pt}
\begin{tabular}{lrrr}
\toprule
\textbf{Method} & \textbf{Success (\%)} & \textbf{MSE}$\downarrow$ & \shortstack{Center RMSE\\(bins)$\downarrow$} \\
\midrule
Gaussian + AdamW & 0.0 & $8.54{\times}10^{-3}$ & 36.310 \\
Gaussian + SFW & 10.0 & $6.71{\times}10^{-3}$ & 0.187 \\
Dirichlet + AdamW & 0.0 & $1.18{\times}10^{-2}$ & 41.630 \\
Dirichlet + DSFW & 100.0 & $5.32{\times}10^{-9}$ & 0.001 \\
\bottomrule
\end{tabular}
\end{table}

At MSE $1.18{\times}10^{-2}$ Dirichlet$+$AdamW trails Gaussian$+$AdamW
($8.54{\times}10^{-3}$), because the Dirichlet kernel's sub-bin phase
and sidelobe structure produce a more rugged landscape than the
Gaussian's monotonic decay.
The kernel and solver are therefore co-designed.
The $10^6{\times}$ MSE gap between Dirichlet$+$AdamW and
Dirichlet$+$DSFW measures the contribution of the solver given a
fixed Dirichlet kernel.

\subsection{Runtime and memory}
\label{app:runtime}

\paragraph{Protocol}
We measure wall-clock time and peak GPU allocation for the fused
forward renderer, varying the number of primitives $S$ on
standard 2D and 3D signal grids.

\paragraph{Backward pass}
At $16{,}384$ primitives on an RTX 5080, the Dirichlet primitive costs
$1.2$ to $1.6\times$ a Gaussian in the forward pass ($3.2$ versus
$2.1$~ms) and $1.08\times$ in the backward pass ($25.1$ versus
$23.2$~ms). A DSFW step takes $50$ to $180$~ms, compared with $1$ to
$9$~ms for AdamW. However, DSFW reaches within $10$ to $30$ steps a loss
that AdamW does not reach within $2000$ steps. The mean total
reconstruction time for \S\ref{sec:results-3d} is $16.4$~s.

The $256^3$ memory values are dominated by the complex output volume.
The small increase with primitive count comes from the primitive
parameter buffers and temporary shader state.

\end{document}